\documentclass{article}
\usepackage{iclr2027_conference,times}

\usepackage{iftex}
\ifPDFTeX
	\usepackage[utf8]{inputenc}
	\usepackage[T1]{fontenc}
\else
	\usepackage{fontspec}
\fi

\usepackage{amsmath,amsfonts,bm}

\def\eqref#1{equation~\ref{#1}}
\def\Eqref#1{Equation~\ref{#1}}

\def\1{\bm{1}}

\DeclareMathAlphabet{\mathsfit}{\encodingdefault}{\sfdefault}{m}{sl}
\SetMathAlphabet{\mathsfit}{bold}{\encodingdefault}{\sfdefault}{bx}{n}

\usepackage{xcolor}
\definecolor{darkblue}{rgb}{0,0,0.5}
\usepackage{url}
\usepackage{graphicx}
\usepackage{booktabs}
\usepackage{array}
\usepackage{float}
\usepackage{placeins}
\usepackage{amsthm}
\usepackage{etoolbox}
\usepackage{microtype}
\usepackage{xspace}
\usepackage{amssymb}
\usepackage{pgfplots}
\pgfplotsset{compat=1.18}
\usepackage{hyperref}
\hypersetup{
	colorlinks=true,
	citecolor=darkblue,
	linkcolor=darkblue,
	urlcolor=darkblue
}
\makeatletter
\patchcmd{\@maketitle}
	{{\LARGE\sc \@title\par}}
	{{\Large\bfseries\centering \@title\par}}
	{}{}
\renewenvironment{abstract}
	{\vskip.075in\centerline{\large\bfseries Abstract}\vspace{0.5ex}\begin{quote}}
	{\par\end{quote}\vskip 1ex}
\renewcommand\section{\@startsection{section}{1}{\z@}
	{-2.0ex plus -0.5ex minus -0.2ex}
	{1.3ex plus 0.3ex minus 0.2ex}
	{\large\bfseries\raggedright}}
\renewcommand\subsection{\@startsection{subsection}{2}{\z@}
	{-1.8ex plus -0.5ex minus -0.2ex}
	{0.8ex plus 0.2ex}
	{\normalsize\bfseries\raggedright}}
\renewcommand\subsubsection{\@startsection{subsubsection}{3}{\z@}
	{-1.5ex plus -0.5ex minus -0.2ex}
	{0.5ex plus 0.2ex}
	{\normalsize\bfseries\raggedright}}
\renewcommand\paragraph{\@startsection{paragraph}{4}{\z@}
	{1.5ex plus 0.5ex minus 0.2ex}
	{-1em}
	{\normalsize\bfseries}}
\makeatother

\newtheoremstyle{paperplain}
	{6pt}{6pt}
	{\itshape}{}
	{\bfseries}{.}
	{0.5em}{}
\theoremstyle{paperplain}
\newtheorem{theorem}{Theorem}[section]
\newtheorem{proposition}[theorem]{Proposition}
\newtheorem{lemma}[theorem]{Lemma}
\newtheorem{corollary}[theorem]{Corollary}
\theoremstyle{definition}

\newcolumntype{L}[1]{>{\raggedright\arraybackslash}p{#1}}

\newcommand{\mesh}{\mathcal{M}_n}
\newcommand{\torus}{\mathcal{T}_n}
\newcommand{\lattice}{\Lambda_n}
\newcommand{\hexnorm}[1]{\left\lVert #1\right\rVert_{\mathrm{H}}}
\newcommand{\VCzero}{\mathrm{VC}_0}
\newcommand{\VCone}{\mathrm{VC}_1}

\newcommand{\dir}[1]{d_{#1}}
\newcommand{\HU}{H_{\mathrm{U}}}
\newcommand{\HL}{H_{\mathrm{L}}}
\newcommand{\RM}{\mathcal{R}_{\mathrm M}}
\newcommand{\RT}{\mathcal{R}_{\mathrm T}}
\newcommand{\bestlifts}{\Lambda^{\star}}
\newcommand{\bestlift}{\lambda^{\star}}
\newcommand{\RThat}{\widehat{\mathcal{R}}_{\mathrm T}}
\newcommand{\VCmap}{\mathcal{V}}

\pgfplotsset{
	hexaxis/.style={
		width=0.92\linewidth,
		height=0.38\linewidth,
		grid=both,
		tick label style={font=\small},
		label style={font=\small},
		legend style={font=\scriptsize, cells={anchor=west}, draw=none, fill=none},
	},
	hexlegendtop/.style={
		legend style={
			at={(0.5,1.06)},
			anchor=south,
			legend columns=3,
			font=\scriptsize,
			cells={anchor=west},
			draw=none,
			fill=none,
		},
	},
	hexline/.style={mark=*, mark size=1.2pt, line width=0.8pt},
}

\pgfplotstableread{
	rate mesh_fixed mesh_random mesh_credit torus_fixed torus_random torus_credit
	0.005 20.361326 20.4216143 20.418537 16.7192233 16.7193983 16.7081073
	0.01 20.3656523 20.436286 20.4780067 16.7309053 16.7377557 16.706888
	0.02 20.488536 20.531779 20.4869823 16.8007913 16.8338467 16.767166
	0.04 20.699772 20.7618077 20.727259 16.9838777 17.0310713 16.974774
	0.08 21.2441307 21.473465 21.2544623 17.4624493 17.552743 17.380412
	0.12 22.0530473 22.4408033 21.9708157 18.0550877 18.241528 17.932983
	0.16 23.2727733 23.8902677 23.06078 18.8795737 19.26055 18.614594
	0.2 25.2082087 26.4473067 24.7104903 55.028915 77.3902813 19.6331977
	0.3 157.56573 209.801974 95.048623 59.2018147 80.546884 81.3911483
	0.4 206.474978 215.788817 266.931263 54.5954793 67.1179997 56.1799733
	0.5 234.813235 218.240217 236.954223 50.8365343 141.009077 63.9931113
	0.7 745.171835 980.793632 241.411005 615.795144 717.649582 505.222528
	1.0 1546.96003 1824.00505 1590.42413 987.227109 1048.39701 766.058049
}\routinglat

\pgfplotstableread{
	rate mesh_fixed mesh_random mesh_credit torus_fixed torus_random torus_credit
	0.005 0.00255187377 0.00254852071 0.00254852071 0.00255364892 0.00255424063 0.00255700197
	0.01 0.00501065089 0.00501065089 0.00500887574 0.00501203156 0.00501952663 0.00501183432
	0.02 0.0100289941 0.010013215 0.0100197239 0.0100323471 0.0100400394 0.0100167653
	0.04 0.0200102564 0.0199984221 0.020029783 0.0200195266 0.0199915187 0.0199629191
	0.08 0.0399487179 0.0398928994 0.0399311637 0.0399641026 0.039938856 0.0399242604
	0.12 0.0597601578 0.0597491124 0.0598313609 0.0597830375 0.0597715976 0.0598686391
	0.16 0.0797526627 0.079642998 0.0797366864 0.0797852071 0.0798960552 0.0797925049
	0.2 0.099678501 0.0996151874 0.099818146 0.0875627219 0.0875400394 0.0995978304
	0.3 0.125402367 0.123936686 0.1462 0.109585404 0.111124063 0.116845168
	0.4 0.154952071 0.148732939 0.162168245 0.138060947 0.14142071 0.144476923
	0.5 0.183639842 0.180081262 0.186188166 0.167128797 0.153194083 0.17436075
	0.7 0.163082051 0.150474556 0.245586785 0.0574771203 0.078416568 0.11740217
	1.0 0.152441617 0.12446213 0.172487968 0.0330284024 0.0622857988 0.0779587771
}\routingthr

\pgfplotstableread{
	rate meshxy_lat hexmesh_credit_lat hextorus_credit_lat meshxy_thr hexmesh_credit_thr hextorus_credit_thr
	0.005 24.1236863 20.418537 16.7081073 0.00255088757 0.00254852071 0.00255700197
	0.01 24.0598007 20.4780067 16.706888 0.00500946746 0.00500887574 0.00501183432
	0.02 24.2807407 20.4869823 16.767166 0.0100238659 0.0100197239 0.0100167653
	0.04 24.543974 20.727259 16.974774 0.020004142 0.020029783 0.0199629191
	0.08 25.5945203 21.2544623 17.380412 0.039929783 0.0399311637 0.0399242604
	0.12 27.2514937 21.9708157 17.932983 0.0597268245 0.0598313609 0.0598686391
	0.16 30.6321577 23.06078 18.614594 0.0796968442 0.0797366864 0.0797925049
	0.2 47.8787473 24.7104903 19.6331977 0.099456213 0.099818146 0.0995978304
	0.3 210.471369 95.048623 81.3911483 0.128015976 0.1462 0.116845168
	0.4 195.470169 266.931263 56.1799733 0.153678107 0.162168245 0.144476923
	0.5 231.42108 236.954223 63.9931113 0.181240039 0.186188166 0.17436075
	0.7 934.772082 241.411005 505.222528 0.185278501 0.245586785 0.11740217
	1.0 1514.71179 1590.42413 766.058049 0.170942012 0.172487968 0.0779587771
}\baselineperf

\pgfplotstableread{
	rate mesh_fixed mesh_credit torus_fixed torus_credit
	0.08 1.00072 1.00925 1.00372 1.02245
	0.2 0.724957 1.70944 1.26518 2.42412
	0.3 1.43483 2.46487 1.21417 1.19165
	0.4 1.45043 1.24224 1.35095 1.43175
	0.5 1.01325 1.29763 1.25080 1.60433
	0.7 1.18800 2.20508 1.27804 1.52267
}\sectorratio

\pgfplotstableread{
	rate uniform_vc1 dateline_vc1 uniform_dlstall dateline_dlstall
	0.04 0.277771202 0.762124774 791.666667 4309.33333
	0.08 0.280185982 0.761530118 3176 23247.6667
	0.12 0.279119886 0.762415212 8409 89016.3333
	0.16 0.278725848 0.762651717 16857 244645
	0.2 0.279153423 0.767608411 30432 539903.333
	0.3 0.285049292 0.776078888 29110.6667 1004413.67
	0.4 0.287910803 0.776861414 42531.3333 1437239.67
}\vcdata

\pgfplotstableread{
	n mesh_fixed_zll mesh_credit_zll torus_fixed_zll torus_credit_zll mesh_fixed_thr mesh_credit_thr torus_fixed_thr torus_credit_thr
	4 13.0160593 12.9766803 11.263146 11.253468 0.383627027 0.402461261 0.400230631 0.401238739
	8 20.361326 20.418537 16.7192233 16.7081073 0.183639842 0.245586785 0.167128797 0.17436075
	12 27.6997913 27.6271993 22.0668273 22.071546 0.141655668 0.147676071 0.101896138 0.103951553
}\scalingdata

\title{Minimal Deadlock-Free Routing for Degree-Six\\
	Triangular-Lattice Meshes and Tori with Two Forbidden Turns}

\author{
	\textbf{Zibo Diao}\thanks{These authors contributed equally to this work.} \\
	\normalfont
	IIIS \\
	Tsinghua University \\
	\texttt{
		diaozb25@mails.tsinghua.edu.cn}
	\And
	\textbf{Rongxi Sun}\footnotemark[1] \\
	\normalfont
	IIIS \\
	Tsinghua University \\
	\texttt{srx24@mails.tsinghua.edu.cn}
}
\iclrfinalcopy

\begin{document}
	\maketitle
	
	\begin{abstract}
		Degree-six triangular-lattice interconnection networks offer substantial
		minimal-path diversity, but their additional directions complicate
		deadlock-free routing under wormhole flow control. We study a finite
		hexagon-shaped mesh and its periodic torus quotient in a common
		six-direction coordinate system. For the finite mesh, we construct a
		minimal partially adaptive routing relation that uses one virtual channel
		and forbids only two directed turns. For the torus, we prove that every
		source--destination pair has a unique closest lattice lift, but that the
		same two-turn physical routing relation still has a cyclic one-VC resource
		CDG for every $n\geq3$. We eliminate this residual periodic dependency by
		combining two virtual channels with Hamiltonian coordinates and
		group-specific datelines. Each same-group segment crosses its dateline at
		most once, which permits a global rank on VC-labelled channel resources.
		We prove minimal all-pairs connectivity for both physical routing
		relations and acyclicity of the complete resource CDG for the proposed
		one-VC mesh and two-VC torus constructions. For a single static
		bidirectional link failure known before a routing epoch, we further rotate
		the turn rule toward the failed orientation and replace a failed hop by a
		same-group two-hop triangle bypass. This restricted extension preserves
		all-pairs connectivity and the original VC counts, with at most one
		additional hop relative to the healthy shortest-path distance.
	\end{abstract}
	\section{Introduction}
	
	Scalable many-core and accelerator systems increasingly rely on
	structured interconnection networks in place of global buses and ad hoc
	point-to-point wiring~\citep{benini2002noc,dallytowles2004}. Routing in
	such networks must balance short paths, adaptivity under nonuniform
	traffic, and freedom from protocol deadlock. Under wormhole flow
	control, correctness depends not only on the existence of a path for
	each source--destination pair, but also on whether the union of all
	permitted paths creates cyclic resource dependencies~\citep{dallyseitz1987}.
	
	Most classical routing theory is developed for orthogonal meshes and
	$k$-ary $n$-cubes. The turn model removes selected direction changes to
	break channel-dependency cycles while preserving partial
	adaptivity~\citep{glassni1994}. Virtual channels provide a complementary
	mechanism by splitting a physical channel into independently allocated
	resource classes~\citep{dally1992vc,dallyaoki1993}. A degree-six
	triangular lattice has different elementary cycles and shortest-path
	sectors, and its periodic quotient additionally contains
	non-contractible channel cycles. These differences require routing rules
	and correctness arguments tailored to the topology.
	
	This paper gives a unified treatment of a finite degree-six mesh and its
	periodic torus. Our contributions are:
	\begin{itemize}
		\item We formalize both topologies in one six-direction integer
		coordinate system, characterize lattice geodesics by adjacent-direction
		sectors, and prove that every torus source--destination pair has a
		unique closest lattice lift.
		\item We construct a one-VC mesh routing relation that forbids only
		$\dir{0}\!\rightarrow\!\dir{5}$ and
		$\dir{2}\!\rightarrow\!\dir{3}$ while preserving minimal
		all-pairs connectivity and partial adaptivity, and prove complete-CDG
		acyclicity using a linear potential.
		\item For the corresponding torus physical routing relation, we show
		that one VC still leaves a directed Hamiltonian resource cycle for
		$n\geq3$. We then cut the cyclic dependency order with two VCs and
		Hamiltonian datelines and prove complete resource-CDG acyclicity by a
		global rank function.
		\item We give a theoretical extension for one static, globally known
			bidirectional link failure. An orientation-specific rotation of the
			turn rule and a controlled two-hop triangle bypass preserve all-pairs
			connectivity, add at most one hop relative to a healthy shortest route,
			and retain the one-VC mesh and two-VC torus deadlock guarantees.
		\item We implement the proposed topologies and routing relations in
			gem5/Garnet, validate the routing relations with an independent CDG
			checker, and evaluate routing performance, restricted-sector pressure,
			torus dateline VC behavior, and scaling.
	\end{itemize}
	
	The main text specifies the routing relations and states the formal
	results. Complete proofs are deferred to the appendix, and
	Sections~\ref{sec:implementation}--\ref{sec:evaluation} describe the
	implementation and experimental results.
	
	\section{Related Work}
	\vspace{0.25\baselineskip}
	
	\paragraph{Deadlock, virtual channels, and adaptive routing.}
	Dally and Seitz introduced the channel-dependency graph formulation for
	wormhole routing~\citep{dallyseitz1987}. Virtual-channel flow control
	separates resource classes on a physical link and can be used to embed a
	cyclic physical topology into an acyclic resource order~\citep{dally1992vc}.
	Dally and Aoki developed adaptive deadlock-free schemes based on VC
	classes and direction reversals~\citep{dallyaoki1993}, while Duato's
	theory permits a cyclic adaptive subnetwork when a suitable deadlock-free
	escape subnetwork exists~\citep{duato1993}. In contrast, the guarantees
	in this paper use the stronger certificate that the complete resource
	CDG of each proposed routing relation is acyclic.
	
	\paragraph{Turn-restricted and fault-tolerant routing.}
	The turn model prohibits selected direction changes to break dependency
	cycles while retaining multiple legal paths~\citep{glassni1994}. Our
	mesh construction applies this principle to the six shortest-path
	sectors of a triangular lattice. The torus construction additionally addresses the periodic,
	non-contractible dependencies that survive the two turn restrictions
	used here. Fault-tolerant wormhole routing has also been studied for faulty
	orthogonal meshes~\citep{glassni1993fault}; in particular, odd-even turn
	restrictions have been adapted to provide deadlock-free routing around
	2D-mesh faults~\citep{wu2003fault}. Our fault result is deliberately
	narrower: it treats one statically known link and preserves the analytic
	resource ordering of the healthy triangular-lattice construction.
	
	\paragraph{Degree-six meshes.}
	Hexagonal-network terminology is not uniform: it may describe a
	degree-three honeycomb, a ring- or hierarchy-based network, or the
	degree-six triangular lattice studied here. Physical-design work on
	hexagonal processor tiles established the implementation and locality
	benefits of a six-neighbor on-chip array~\citep{xiao2012hex}, while
	other studies considered hierarchical hexagon topologies and routing
	heuristics~\citep{baleski2015hex,cheng2020hex}. More directly related
	to our mesh, Gu et al. adapted turn-model routing to a corner-addressed
	hexagonal network and obtained minimal partially adaptive
	algorithms~\citep{gu2006hex}. Albader et al. gave Eisenstein--Jacobi
	coordinates and shortest-path communication algorithms for hexagonal
	meshes~\citep{albader2012hex}. Moriam and Fettweis subsequently used a
	channel-dependency matrix to synthesize VC-free, fault-tolerant
	turn-restricted algorithms for a diagonally augmented mesh
	NoC~\citep{moriam2016hex}. These works motivate the topology and the
	turn-model approach, but they do not give the particular two-turn,
	one-VC routing relation and complete analytic CDG certificate proved
	here.
	
	\paragraph{Closest hexagonal-torus routing work.}
	Shamaei et al. proposed the closest prior construction for a hexagonal
	torus family: a minimal, fully adaptive algorithm that partitions shortest
	paths into six message types and assigns three VC classes according to
	message type and whether a route uses a wraparound link~\citep{shamaei2013hex}.
	Appendix~\ref{app:prior-reassessment} gives the coordinate identification
	for the $H_4$ instance used in our finite counterexample. Under the routing
	and VC-assignment rules stated in that paper, our reconstruction contains
	a same-type, same-VC directed resource cycle. Proposition~\ref{prop:prior-cycle}
	states the witness, and Appendix~\ref{app:prior-reassessment} lists every
	packet and an automated check. This contradicts the claimed complete-CDG
	acyclicity under those stated rules; by itself it does not establish that
	a reachable wormhole deadlock exists for the fully adaptive routing
	function. Unlike that three-VC proposal, our positive result intentionally
	gives up full adaptivity, uses two VCs, and proves that the \emph{complete}
	resource CDG is acyclic by a global rank function.
	
	\subsection{Reassessing the prior three-VC scheme}
	
	Let $H_4$ denote the $n=4$ torus instance of
	\citet{shamaei2013hex}, represented in our coordinates by the
	representative set $V_4$ and periods $T_1=(4,3)$ and $T_2=(-3,7)$.
	Write $(c,q)$ for directed physical channel $c$ in VC class $q$.
	
	\begin{proposition}[A permitted cycle in the published three-VC CDG]
		\label{prop:prior-cycle}
		Under the routing and VC-assignment rules stated by
		\citet{shamaei2013hex}, the resource CDG of $H_4$ contains the directed
		cycle
		\begin{equation}
			(c_0,1)\rightarrow(c_1,1)\rightarrow\cdots\rightarrow
			(c_6,1)\rightarrow(c_0,1),
			\label{eq:prior-cycle}
		\end{equation}
		where the projected channels follow
		\begin{equation}
			\begin{split}
				(3,0)&\xrightarrow{c_0}(0,-3)
				\xrightarrow{c_1}(1,-3)
				\xrightarrow{c_2}(2,-3)
				\xrightarrow{c_3}(3,-3)\\
				&\xrightarrow{c_4}(3,-2)
				\xrightarrow{c_5}(3,-1)
				\xrightarrow{c_6}(3,0).
			\end{split}
			\label{eq:prior-channels}
		\end{equation}
		Every dependency in~\eqref{eq:prior-cycle} is induced by a minimal
		wraparound Type-1 route allowed by the published fully adaptive routing
		function.
	\end{proposition}
	
	Appendix~\ref{app:prior-reassessment} gives the proof and clarifies the
	scope of this result. In particular, Proposition~\ref{prop:prior-cycle}
	contradicts the paper's complete-CDG acyclicity conclusion and is not
	excluded by its Theorem~2 argument, but we do not infer packet-level
	deadlock reachability from a CDG cycle alone.
	
	\section{Network Model and Preliminaries}
	\label{sec:model}
	
	\subsection{Six-direction coordinates and the finite mesh}
	
	Nodes of the infinite triangular lattice are integer pairs
	$v=(x,y)\in\mathbb Z^2$. Its six directed unit steps, listed
	counterclockwise, are
	\begin{equation}
		\begin{aligned}
			\dir{0}&=(1,0), & \dir{1}&=(0,1), & \dir{2}&=(-1,1),\\
			\dir{3}&=(-1,0),& \dir{4}&=(0,-1),& \dir{5}&=(1,-1).
		\end{aligned}
		\label{eq:directions}
	\end{equation}
	Indices are interpreted modulo six, and $\dir{i+3}=-\dir{i}$. Define
	the hexagonal norm by
	\begin{equation}
		\hexnorm{(x,y)}=\max\{|x|,|y|,|x+y|\}.
		\label{eq:hexnorm}
	\end{equation}
	
	\begin{lemma}[Adjacent-direction decomposition]
		\label{lem:decomposition}
		For every $\Delta\in\mathbb Z^2$, there exist adjacent directions
		$\dir{i},\dir{i+1}$ and integers $a,b\geq0$ such that
		\begin{equation}
			\Delta=a\dir{i}+b\dir{i+1},
			\qquad a+b=\hexnorm{\Delta}.
		\end{equation}
	\end{lemma}
	
	\begin{lemma}[Triangular-lattice distance]
		\label{lem:lattice-distance}
		For all $s,t\in\mathbb Z^2$, the shortest-path distance in the infinite
		triangular lattice is
		\begin{equation}
			d_{\mathrm H}(s,t)=\hexnorm{t-s}.
		\end{equation}
	\end{lemma}
	
	\begin{lemma}[Direction set of a lattice geodesic]
		\label{lem:geodesic-directions}
		Every shortest path in the infinite triangular lattice uses steps from
		at most two adjacent directions. Equivalently, the direction word of
		every lattice geodesic is contained in one adjacent-direction sector.
	\end{lemma}
	
	Proofs of Lemmas~\ref{lem:decomposition}--\ref{lem:geodesic-directions}
	are given in Appendix~\ref{app:latticefacts}.
	
	For later reference, define the six closed direction sectors
	\begin{equation}
		S_i=\{a\dir{i}+b\dir{i+1}:a,b\geq0\},
		\qquad
		S_i^\circ=\{a\dir{i}+b\dir{i+1}:a,b>0\}.
		\label{eq:sectors}
	\end{equation}
	The strict interior $S_i^\circ$ excludes the two boundary rays. The four
	sectors $S_0,S_1,S_3,S_4$ are \emph{within-group sectors}, while
	$S_2$ and $S_5$ are \emph{cross-group sectors} for the direction groups
	introduced in Section~\ref{sec:routing}. For continuity with the
	experimental labels, Section~\ref{sec:evaluation} refers to these two
	classes as internal and boundary sectors, respectively.
	
	For a size parameter $n\geq2$, let $R=n-1$ and define the finite mesh
	$\mesh=(V_n,E_n)$ by
	\begin{equation}
		V_n=\{(x,y)\in\mathbb Z^2:\hexnorm{(x,y)}\leq R\}.
		\label{eq:meshvertices}
	\end{equation}
	Whenever $u,u+\dir{i}\in V_n$, the graph contains a directed channel in
	each direction between the two nodes. Its node count is
	\begin{equation}
		|V_n|=1+6\sum_{r=1}^{n-1}r=3n^2-3n+1.
		\label{eq:nodecount}
	\end{equation}
	
	\begin{figure}[t]
		\centering
		\includegraphics[width=0.96\linewidth]{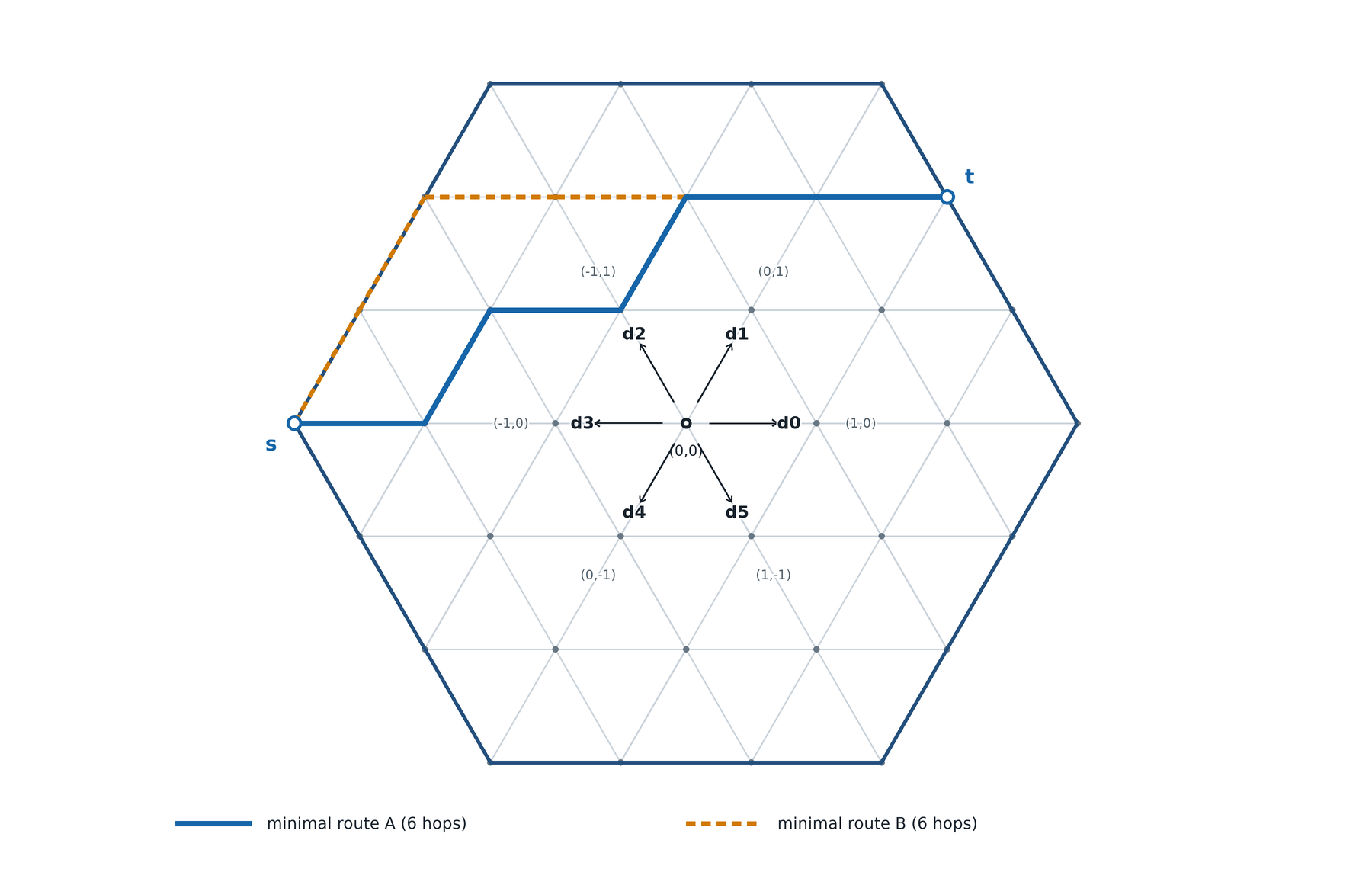}
		\caption{Coordinate system and minimal paths in a finite degree-six
			triangular-lattice mesh with $n=4$ ($R=3$). The highlighted routes are
			two distinct six-hop minimal paths from $s=(-3,0)$ to $t=(1,2)$.}
		\label{fig:hex-mesh}
	\end{figure}
	
	Figure~\ref{fig:hex-mesh} illustrates the directed link orientations,
	the finite boundary, and the minimal-path diversity represented by
	Lemma~\ref{lem:decomposition}.
	
	\subsection{The periodic torus}
	
	Define the period vectors
	\begin{equation}
		T_1=(n,n-1),\qquad T_2=(-(n-1),2n-1),
		\label{eq:periods}
	\end{equation}
	and the period lattice
	$\lattice=\{aT_1+bT_2:a,b\in\mathbb Z\}$. The torus is the quotient
	\begin{equation}
		\torus=\mathbb Z^2/\lattice.
		\label{eq:torus}
	\end{equation}
	Coordinates $p$ and $p+\lambda$ represent the same physical node for
	all $\lambda\in\lattice$. Since
	\begin{equation}
		N=|\det(T_1,T_2)|=3n^2-3n+1,
		\label{eq:N}
	\end{equation}
	the quotient contains $N$ nodes. A directed torus channel is
	$[u]\rightarrow[u+\dir{i}]$, where brackets denote an equivalence
	class. Its distance is
	\begin{equation}
		d_{\torus}([s],[t])=
		\min_{\lambda\in\lattice}\hexnorm{t-s+\lambda}.
		\label{eq:torusdistance}
	\end{equation}
	
	\begin{lemma}[Complete representative set]
		\label{lem:representatives}
		The $N$ nodes in $V_n$ form a complete, nonredundant set of
		representatives for $\mathbb Z^2/\lattice$.
	\end{lemma}
	
	\begin{corollary}[Torus diameter]
		\label{cor:diameter}
		The diameter of $\torus$ is $n-1$.
	\end{corollary}
	
	\begin{corollary}[Unique closest lift]
		\label{cor:unique-lift}
		For every ordered pair $[s],[t]\in\torus$, with their canonical
		representatives $s,t\in V_n$, there is a unique
		$\bestlift(s,t)\in\lattice$ such that
		\begin{equation}
			d_{\torus}([s],[t])
			=\hexnorm{t-s+\bestlift(s,t)}.
			\label{eq:unique-lift-distance}
		\end{equation}
	\end{corollary}
	
	Proofs of Lemma~\ref{lem:representatives} and
	Corollaries~\ref{cor:diameter}--\ref{cor:unique-lift} are given in
	Appendix~\ref{app:torusgeometry}. We retain the set notation
	\begin{equation}
		\bestlifts(s,t)
		:=\operatorname*{arg\,min}_{\lambda\in\lattice}
		\hexnorm{t-s+\lambda}
		=\{\bestlift(s,t)\}
		\label{eq:closest-lifts}
	\end{equation}
	when convenient. Thus every minimal torus route has the same destination
	lift $t+\bestlift(s,t)$; its remaining freedom is only the ordering of
	minimal direction steps within that lift.
	
	\begin{figure}[t]
		\centering
		\includegraphics[width=0.98\linewidth]{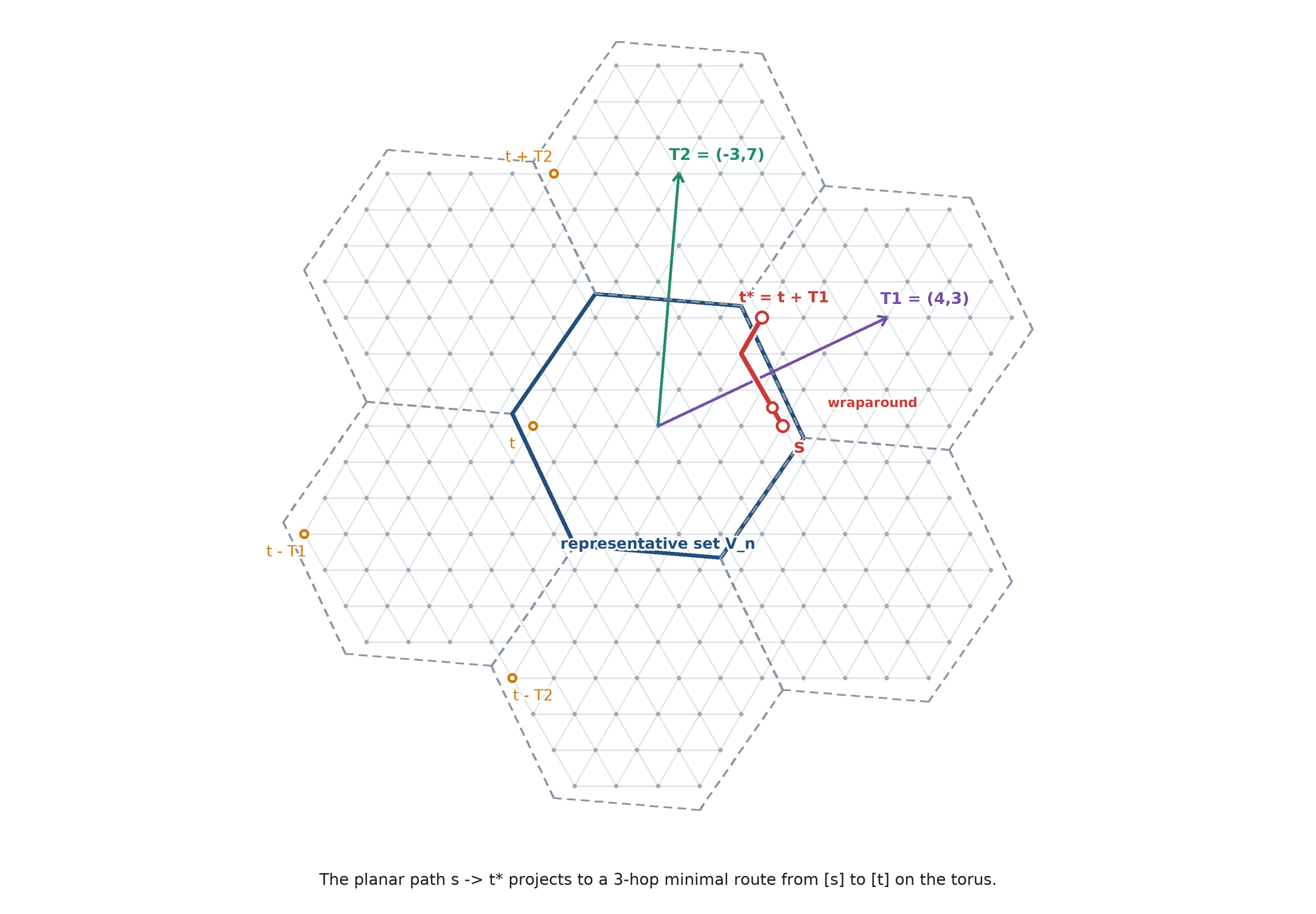}
		\caption{Planar lifting of a minimal route in the $n=4$ torus. The solid
			central hexagon depicts the representative node set $V_n$, and the
			dashed copies are its translates under $T_1$ and $T_2$. The continuous
			boundaries are visualization aids; quotient representatives are the
			discrete nodes in $V_n$. The destination $t$ and its marked translates
			represent the same torus node. Routing from $s=(3,0)$ to the unique
			closest lift $t^*=t+T_1=(1,3)$ gives a three-hop lattice geodesic whose
			projection is a minimal wraparound route from $[s]$ to $[t]$.}
		\label{fig:torus-lift}
	\end{figure}
	
	The geometric boundary crossing in Figure~\ref{fig:torus-lift}
	illustrates wraparound connectivity. It is distinct from the Hamiltonian
	dateline defined in Section~\ref{sec:torus-routing}.
	
	\subsection{Routing relations and channel dependencies}
	
	For a directed path $P=(u_0,\ldots,u_m)$, let
	$c_j:u_j\rightarrow u_{j+1}$ denote its $j$th physical channel. Each
	channel carries a direction label $\dir{i_j}$; the corresponding
	direction word is $\dir{i_0}\dir{i_1}\cdots\dir{i_{m-1}}$. We say that
	$P$ contains the directed turn $\dir{i}\rightarrow\dir{k}$ when
	$\dir{i}\dir{k}$ occurs as a consecutive subword.
	
	A \emph{path-based routing relation} $\mathcal R$ assigns to every
	ordered source--destination pair $(s,t)$ a set $\mathcal R(s,t)$ of
	permitted directed paths. It is \emph{minimal} if every permitted path
	has shortest-path length, and \emph{all-pairs connected} if
	$\mathcal R(s,t)\neq\varnothing$ for every ordered pair. It is
	\emph{partially adaptive} if it is all-pairs connected and there exist
	$s,t$ and a permitted route prefix for which at least two distinct next
	channels extend the prefix to paths in $\mathcal R(s,t)$.
	
	At each hop, an implementation may choose any output that extends the
	current prefix to a path in the relevant routing relation. The choice
	may be random or may use local congestion information. A deterministic
	implementation may instead apply a fixed priority order; the correctness
	results apply to every such policy because they are proved for the union
	of all permitted paths.
	
	We assume standard wormhole flow control, no U-turns, and a finite set
	of VCs on each directed physical channel. A \emph{resource} is a pair
	$(c,q)$ consisting of a directed physical channel $c$ and VC index $q$.
	A VC-labelled path has a resource word
	\begin{equation}
		(c_0,q_0),(c_1,q_1),\ldots,(c_{m-1},q_{m-1}).
		\label{eq:resource-word}
	\end{equation}
	The \emph{complete resource channel-dependency graph} has one vertex for
	every resource used by a permitted path and contains the directed edge
	\begin{equation}
		(c_j,q_j)\longrightarrow(c_{j+1},q_{j+1})
		\label{eq:cdg-edge}
	\end{equation}
	whenever the two resources occur consecutively in the resource word of
	some permitted path. Longer hold-and-wait dependency chains are
	represented by directed paths in this graph. We use the standard
	acyclicity certificate: a VC-labelled routing relation whose complete
	resource CDG is acyclic is deadlock-free~\citep{dallyseitz1987}.
	
	\section{Deadlock-Free Routing}
	\label{sec:routing}
	
	Partition the directions into the two half-plane groups
	\begin{equation}
		\mathrm U=\{\dir{0},\dir{1},\dir{2}\},\qquad
		\mathrm L=\{\dir{3},\dir{4},\dir{5}\}.
		\label{eq:groups}
	\end{equation}
	Both constructions prohibit the directed turns
	\begin{equation}
		\boxed{\dir{0}\rightarrow\dir{5},\qquad
			\dir{2}\rightarrow\dir{3}}.
		\label{eq:bannedturns}
	\end{equation}
	By Lemma~\ref{lem:geodesic-directions}, every minimal direction word is
	confined to one adjacent-direction sector. Within this class of minimal
	words, avoiding the two turns in~\eqref{eq:bannedturns} is equivalent to
	the group word belonging to the language $\mathrm L^*\mathrm U^*$.
	
	\subsection{Mesh routing with two forbidden turns and one VC}
	\label{sec:mesh-routing}
	
	For $s,t\in V_n$, define
	\begin{equation}
		\begin{aligned}
			\RM(s,t)=\{P:\;&P\text{ is an $s$--$t$ path of length }d_{\mathrm H}(s,t),\\
			& P\text{ contains neither prohibited turn in equation~\ref{eq:bannedturns}}\}.
		\end{aligned}
		\label{eq:mesh-relation}
	\end{equation}
	All channels in every $P\in\RM(s,t)$ use the single available VC.
	Operationally, the router may choose a minimal output only when the
	resulting prefix remains extendible to a path in $\RM(s,t)$. Thus the
	choice may be randomized at every hop or selected using local congestion.
	In the cross-group sector $S_5=\{\dir{5},\dir{0}\}$, all $\dir{5}$ hops
	precede all $\dir{0}$ hops; in $S_2=\{\dir{2},\dir{3}\}$, all $\dir{3}$
	hops precede all $\dir{2}$ hops. The four within-group sectors retain
	every interleaving of their two minimal directions.
	
	\begin{theorem}[Mesh routing]
		\label{thm:mesh}
		For every $n\geq2$, $\RM$ is minimal, all-pairs connected, partially
		adaptive, and deadlock-free using one VC. If
		$t-s=a\dir{i}+b\dir{i+1}\in S_i^\circ$, then $\RM(s,t)$ retains all
		$\binom{a+b}{a}$ direction interleavings when $S_i$ is a within-group
		sector, whereas it retains only the unique $\mathrm L$-before-$\mathrm U$
		word when $S_i$ is a cross-group sector. On a sector boundary ray,
		the minimal direction word is unique.
	\end{theorem}
	
	The complete proof, including minimal all-pairs connectivity inside the
	finite boundary, is given in Appendix~\ref{app:meshproof}.
	
	\subsection{Torus routing with two forbidden turns and two VCs}
	\label{sec:torus-routing}
	
	For quotient nodes $[s],[t]\in\torus$, let $s,t\in V_n$ denote their
	canonical representatives and let $\bestlift(s,t)$ be the unique closest
	lift vector from Corollary~\ref{cor:unique-lift}. Define the physical
	torus routing relation
	\begin{equation}
		\begin{split}
			\RT([s],[t])=\{\pi(P):\ &P\text{ is a lattice geodesic from }s
			\text{ to }t+\bestlift(s,t),\\
			&\text{and }P\text{ contains neither prohibited turn
				in~\eqref{eq:bannedturns}}\},
		\end{split}
		\label{eq:torus-relation}
	\end{equation}
	where $\pi$ denotes projection to the quotient. The displacement
	$t-s+\bestlift(s,t)$ admits an adjacent-direction decomposition from
	Lemma~\ref{lem:decomposition}. If both coefficients are positive, the
	route uses those two directions; on a sector ray only one direction has
	positive count. At each hop, any legal direction may be chosen provided
	that the resulting prefix remains extendible to a path in $\RT$.
	
	Let
	\begin{equation}
		N=3n^2-3n+1,\qquad k=3n-1,
		\label{eq:Nk}
	\end{equation}
	and define quotient coordinates with residues in
	$\{0,\ldots,N-1\}$:
	\begin{equation}
		\HU(x,y)=x+ky\pmod N,\qquad
		\HL(x,y)=-x-ky\pmod N.
		\label{eq:hamcoords}
	\end{equation}
	Their direction increments are
	\begin{equation}
		\begin{array}{c|ccc}
			& \text{first direction} & \text{second direction} & \text{third direction}\\ \hline
			\mathrm U & \Delta\HU(\dir{0})=1 & \Delta\HU(\dir{1})=k & \Delta\HU(\dir{2})=k-1\\
			\mathrm L & \Delta\HL(\dir{3})=1 & \Delta\HL(\dir{4})=k & \Delta\HL(\dir{5})=k-1
		\end{array}
		\label{eq:increments}
	\end{equation}
	
	\begin{lemma}[Hamiltonian coordinates]
		\label{lem:hamiltonian}
		The maps $\HU,\HL:\torus\rightarrow\mathbb Z_N$ are well-defined
		bijections. In particular, repeated $\dir{0}$ channels follow the cyclic
		order induced by $\HU$, and repeated $\dir{3}$ channels follow the
		cyclic order induced by $\HL$.
	\end{lemma}
	
	\begin{proposition}[Residual one-VC cycle on the torus]
		\label{prop:torus-one-vc-cycle}
		For every $n\geq3$, if every hop permitted by $\RT$ is assigned the
		same VC, the complete resource CDG is cyclic. In particular, with
		$c_r:[r\dir{0}]\rightarrow[(r+1)\dir{0}]$ for
		$r\in\mathbb Z_N$, it contains
		\begin{equation}
			(c_0,0)\rightarrow(c_1,0)\rightarrow\cdots\rightarrow
			(c_{N-1},0)\rightarrow(c_0,0).
			\label{eq:torus-one-vc-cycle}
		\end{equation}
	\end{proposition}
	
	Proofs of Lemma~\ref{lem:hamiltonian} and
	Proposition~\ref{prop:torus-one-vc-cycle} are given in
	Appendices~\ref{app:torusgeometry} and~\ref{app:torusproof},
	respectively. Proposition~\ref{prop:torus-one-vc-cycle} does not claim
	that two VCs are necessary for every possible torus routing relation;
	it shows only that assigning one VC to every hop permitted by
	$\RT$ yields a cyclic complete resource CDG.
	
	For $G\in\{\mathrm L,\mathrm U\}$, let
	\begin{equation}
		H_G=\begin{cases}
			\HL,&G=\mathrm L,\\
			\HU,&G=\mathrm U.
		\end{cases}
		\label{eq:group-coordinate}
	\end{equation}
	A channel $c_j:[u_j]\rightarrow[u_{j+1}]$ in group $G$ is a
	\emph{dateline channel} if $H_G(u_{j+1})<H_G(u_j)$.
	
	For a lifted geodesic $P=(u_0,\ldots,u_m)$ whose projection lies in
	$\RT([s],[t])$, let
	$c_j:[u_j]\rightarrow[u_{j+1}]$ denote the projected physical channel.
	Define its per-hop VC mapping $\VCmap(P,j)\in\{0,1\}$ as follows. At the
	beginning of every maximal same-group segment, initialize
	\texttt{crossed}$=0$. For the $j$th channel in that segment,
	\begin{equation}
		\VCmap(P,j)=
		\begin{cases}
			0, & \texttt{crossed}=0\text{ and }c_j\text{ is not a dateline channel},\\
			1, & \texttt{crossed}=0\text{ and }c_j\text{ is a dateline channel},\\
			1, & \texttt{crossed}=1.
		\end{cases}
		\label{eq:vc-map}
	\end{equation}
	After assigning a dateline channel, set \texttt{crossed}$=1$. A route
	can change groups only from $\mathrm L$ to $\mathrm U$; at that boundary
	the state is reset before assigning the first $\mathrm U$ channel.
	Hence the reset may induce a dependency from
	$(\mathrm L,\VCone)$ to $(\mathrm U,\VCzero)$.
	
	\Eqref{eq:vc-map} lifts the physical relation $\RT$ to the
	VC-labelled resource relation
	\begin{equation}
		\begin{aligned}
			\RThat([s],[t])=\{&\bigl((c_0,\VCmap(P,0)),\ldots,
			(c_{m-1},\VCmap(P,m-1))\bigr):\\
			& P\text{ is a lattice geodesic with }\pi(P)\in\RT([s],[t])\}.
		\end{aligned}
		\label{eq:torus-resource-relation}
	\end{equation}
	
	\paragraph{Per-packet implementation.}
	The routing state consists of the unique closest lift, the remaining hop
	counts, the current direction group, and the bit \texttt{crossed}:
	\begin{enumerate}
		\item At injection, compute and fix $\bestlift(s,t)$; choose an
		adjacent-direction decomposition of its displacement and initialize
		the remaining hop counts, ignoring any zero-count direction.
		\item At each hop, form the legal minimal candidate set from directions
		with positive remaining counts whose choice leaves the prefix
		extendible to a path in $\RT$. Equivalently, in a cross-group sector
		the $\mathrm U$ direction is unavailable while a lower-group hop remains.
		\item Select a candidate randomly, by local congestion, or by a fixed
		priority, then decrement its remaining count.
		\item Assign its VC by~\eqref{eq:vc-map}; update \texttt{crossed}, and
		reset it before the first channel after an
		$\mathrm L\rightarrow\mathrm U$ group transition.
	\end{enumerate}
	The closest lift is fixed at injection, and the only dynamic state used
	specifically by the VC phase rule is the one-bit \texttt{crossed} flag.
	
	\begin{theorem}[Torus routing]
		\label{thm:torus}
		For every $n\geq2$, the physical relation $\RT$ is minimal and
		all-pairs connected. The two-VC resource relation $\RThat$ has an
		acyclic complete resource CDG and is therefore deadlock-free. If
		$t-s+\bestlift(s,t)=a\dir{i}+b\dir{i+1}\in S_i^\circ$, then $\RT$
		retains all $\binom{a+b}{a}$ direction interleavings when $S_i$ is a
		within-group sector, whereas it retains only the unique
		$\mathrm L$-before-$\mathrm U$ word when $S_i$ is a cross-group sector.
		On a sector boundary ray, the minimal direction word is unique.
		Moreover, $\RT$ is partially adaptive if and only if $n\geq3$.
	\end{theorem}
	
	The complete proof is given in Appendix~\ref{app:torusproof}.

	\begin{corollary}[Per-packet VC transition bound]
		\label{cor:vc-transitions}
		Every VC word in $\RThat$ has the form
		$0^*1^*$ or $0^*1^*0^*1^*$. Consequently, a packet changes VC class
		at most three times along its route.
	\end{corollary}

	This bound follows from the $\mathrm L^*\mathrm U^*$ group order and the
	single-dateline property of each maximal same-group segment; a short
	derivation is included in Appendix~\ref{app:torusproof}.

	\subsection{Restricted single-link fault extension}
	\label{sec:single-link-fault}

	We consider one failed bidirectional physical link
	$e=\{A,A+\dir{j}\}$, where $j\in\{0,1,2\}$ selects its undirected
	orientation. In the torus, the endpoints and channels are interpreted in
	the quotient. The failure is static, is globally known before a new
	routing epoch begins, and is the only failed component. All routers use the
	same fault-aware configuration; packets routed under an older
	configuration are not present in that epoch.

	Orient the two direction groups toward the failed link by defining
	\begin{equation}
		\mathrm U_j=\{\dir{j-1},\dir{j},\dir{j+1}\},\qquad
		\mathrm L_j=\{\dir{j+2},\dir{j+3},\dir{j+4}\},
		\label{eq:fault-groups}
	\end{equation}
	and prohibit
	\begin{equation}
		\boxed{\dir{j+1}\rightarrow\dir{j+2},\qquad
			\dir{j-1}\rightarrow\dir{j-2}}.
		\label{eq:fault-turns}
	\end{equation}
	Indices remain modulo six. Let $\mathcal R_{X}^{(j)}$, for
	$X\in\{\mathrm M,\mathrm T\}$, denote the corresponding rotated healthy
	routing relation: it is defined exactly as $\RM$ or $\RT$, respectively,
	with~\eqref{eq:fault-turns} replacing~\eqref{eq:bannedturns}.

	For $P\in\mathcal R_X^{(j)}$, let $\mathcal B_e(P)$ be the set obtained by
	leaving $P$ unchanged when it avoids $e$, and otherwise replacing its
	failed directed hop by every geometrically available substitution
	\begin{equation}
		\begin{aligned}
			\dir{j}&\rightsquigarrow
			\dir{j-1}\dir{j+1}
			\quad\text{or}\quad
			\dir{j+1}\dir{j-1},\\
			\dir{j+3}&\rightsquigarrow
			\dir{j+2}\dir{j+4}
			\quad\text{or}\quad
			\dir{j+4}\dir{j+2}.
		\end{aligned}
		\label{eq:fault-bypass}
	\end{equation}
	A mesh substitution is available only when its intermediate node lies in
	$V_n$; both substitutions are available on the torus. Define the physical
	fault-aware relation by
	\begin{equation}
		\mathcal R_{X,e}^{(j)}(s,t)
		=\bigcup_{P\in\mathcal R_X^{(j)}(s,t)}\mathcal B_e(P),
		\qquad X\in\{\mathrm M,\mathrm T\}.
		\label{eq:fault-relation}
	\end{equation}
	Thus both bypass orders remain available whenever both are geometric
	paths. Once the first bypass channel is selected, the second is fixed by
	\eqref{eq:fault-bypass}; this controlled substitution is part of the
	routing relation rather than an arbitrary nonminimal detour.

	For the torus, let $R_{60}(x,y)=(-y,x+y)$ and
	$Q_j=R_{60}^{\,1-j}$. The rotated group coordinates are
	\begin{equation}
		H_{\mathrm U,j}=\HU\circ Q_j,\qquad
		H_{\mathrm L,j}=\HL\circ Q_j.
		\label{eq:fault-hamcoords}
	\end{equation}
	Dateline channels and per-hop VC assignments are defined by
	Equations~\ref{eq:group-coordinate}--\ref{eq:vc-map}, with
	$\mathrm U,\mathrm L$ and $\HU,\HL$ replaced by their rotated versions.
	The state is reset only at an $\mathrm L_j\rightarrow\mathrm U_j$
	transition. Let $\widehat{\mathcal R}_{\mathrm T,e}^{(j)}$ denote the
	resulting VC-labelled relation.

	\begin{proposition}[Restricted single-link fault extension]
		\label{prop:single-link-fault}
		For every $n\geq2$ and every single static bidirectional link failure
		described above, the physical relations
		$\mathcal R_{\mathrm M,e}^{(j)}$ and
		$\mathcal R_{\mathrm T,e}^{(j)}$ are all-pairs connected. Every
		permitted route has length at most one hop greater than the corresponding
		healthy shortest-path distance. The complete one-VC resource CDG of
		$\mathcal R_{\mathrm M,e}^{(j)}$ is acyclic, and the complete resource
		CDG of $\widehat{\mathcal R}_{\mathrm T,e}^{(j)}$ is acyclic using two
		VCs. Hence the finite mesh and periodic torus remain deadlock-free with
		the same respective VC counts under this restricted fault model.
	\end{proposition}

	The complete proof of Proposition~\ref{prop:single-link-fault} is given in
	Appendix~\ref{app:single-link-fault}.
	
	\section{Implementation}
	\label{sec:implementation}
	
	We implemented both topologies and routing relations in gem5
	v23.0.0.1 using Garnet standalone synthetic traffic. The HexMesh topology
	instantiates the vertex set $V_n$ from~\eqref{eq:meshvertices}, with
	$3n^2-3n+1$ routers and bidirectional links in the six directions
	$\dir{0},\ldots,\dir{5}$. The HexTorus topology uses the same
	representatives $V_n$ and adds the quotient links induced by the period
	vectors $T_1=(n,n-1)$ and $T_2=(-(n-1),2n-1)$.
	
	The mesh router implements the one-VC relation $\RM$: it forms the
	minimal candidate set for the remaining displacement and removes any
	candidate that would introduce either prohibited turn
	$\dir{0}\rightarrow\dir{5}$ or $\dir{2}\rightarrow\dir{3}$. The torus
	router first selects a closest destination lift from $\bestlifts(s,t)$,
	decomposes the lifted displacement into adjacent minimal directions, and
	applies the same L$^*$U$^*$ turn restriction. For each torus hop, the VC
	allocator evaluates the Hamiltonian coordinate of the current group,
	assigns VC1 on and after a dateline crossing in that same-group segment,
	and resets the phase before the first U-group hop after an L-to-U group
	transition. Thus the implementation realizes the per-hop VC mapping in
	\eqref{eq:vc-map}; using a single packet-level VC would not be equivalent
	to the algorithm proved in Theorem~\ref{thm:torus}.
	
	On top of the legal minimal candidate set, we compare three selection
	policies. \emph{Fixed} uses a deterministic direction priority.
	\emph{Random} chooses uniformly among the legal candidates. \emph{Credit}
	chooses the candidate whose downstream output has the largest available
	credit count. These policies affect performance only; deadlock freedom
	comes from the routing relation and the torus VC assignment.
	
	We also implemented an independent channel-dependency-graph checker that
	does not import gem5. It constructs the complete resource CDG for the
	mesh and torus relations, including the torus dateline VC rule, and checks
	acyclicity for the tested finite sizes. The simulator records the usual
	Garnet latency and throughput statistics as well as Hex-specific counters
	for direction hops, legal candidate counts, credit rechoices, torus VC0
	and VC1 hops, dateline crossings, group resets, VC stalls, and dateline
	VC stalls.
	
	\section{Evaluation}
	\label{sec:evaluation}
	
	\subsection{Methodology}
	
	Table~\ref{tab:evaluation} summarizes the experiments. Unless otherwise
	stated, each point is averaged over seeds $1,2,3$ and each run executes
	for 10,000 network-tester cycles. Uniform experiments use an all-node
	pair-list generator rather than gem5's built-in uniform-random traffic,
	so that all routers are active sources and the same traffic machinery is
	used across mesh and torus. We report average packet latency, accepted
	packet throughput, zero-load latency at the lowest injection rate,
	maximum observed accepted throughput, hop count, VC1 hop fraction,
	dateline crossings, and VC stall counters. All formal result directories
	were checked for nonempty statistics, normal simulator exit, finite
	latency and throughput, and absence of fatal, panic, assertion, or
	deadlock messages.
	
	\begin{table}[H]
		\caption{Experimental protocol.}
		\label{tab:evaluation}
		\centering
		\small
		\begin{tabular}{@{}L{0.18\linewidth}L{0.25\linewidth}L{0.34\linewidth}L{0.14\linewidth}@{}}
			\toprule
			\textbf{Experiment} & \textbf{Fixed settings} & \textbf{Comparison} & \textbf{Runs}\\
			\midrule
			Routing performance & HexMesh/HexTorus, $n=8$ & fixed, random, and credit selection under all-node uniform pair-list traffic & 234\\
			Conventional baseline & $13\times13$ Mesh\_XY, 169 nodes & deterministic XY routing under the same all-node uniform pair-list traffic & 39\\
			Sector pressure & HexMesh/HexTorus, $n=8$ & boundary-sector versus internal-sector traffic with matched distance distributions & 864\\
			VC behavior & HexTorus, $n=8$ & uniform versus dateline-heavy traffic; VC0/VC1 usage and dateline stalls & 270\\
			Scaling & HexMesh/HexTorus, $n=4,8,12$ & fixed versus credit selection; zero-load latency and saturation throughput & 468\\
			\bottomrule
		\end{tabular}
	\end{table}
	
	The sector traffic generator occasionally relaxes two Mesh boundary
	sources when the strict boundary-sector and distance constraints leave no
	candidate pair for that source. This affects a negligible number of
	pairs and does not indicate a routing or simulator error.
	
	\subsection{Routing performance}
	\label{sec:routing-performance}
	
	Figure~\ref{fig:routing-performance} compares the three candidate
	selection policies at $n=8$. At low injection rates, HexTorus has lower
	packet latency than HexMesh: the zero-load latency is about $16.7$
	cycles for torus and about $20.4$ cycles for mesh. This follows the
	shorter average torus distance under the quotient geometry.
	
	At high load, the credit policy substantially improves mesh throughput.
	Mesh fixed and mesh random peak at accepted packet throughputs of
	$0.1836$ and $0.1801$, respectively, while mesh credit reaches
	$0.2456$, a $1.34\times$ improvement over fixed. Torus credit provides a
	smaller gain, increasing the maximum accepted throughput from $0.1671$
	for fixed to $0.1744$. Thus the remaining minimal adaptivity in the
	proposed relation is practically useful, especially in the finite mesh.
	Torus does not universally dominate throughput under this traffic even
	though it has lower zero-load latency; the shorter paths also concentrate
	contention differently under all-node uniform pair-list injection.
	Table~\ref{tab:routing-summary} gives the corresponding summary values.
	The latency-at-maximum-throughput column is intentionally reported
	because the highest accepted throughput is often already in the onset of
	congestion; it should not be read as a low-latency operating point.
	
	\begin{table}[tbp]
		\caption{Routing performance summary for all-node uniform pair-list traffic at $n=8$.}
		\label{tab:routing-summary}
		\centering
		\small
		\begin{tabular}{@{}llrrrr@{}}
			\toprule
			Topology & Selection & ZLL & Avg. hops & Max thr. & Lat. at max thr.\\
			\midrule
			Mesh & Fixed & 20.36 & 6.81 & 0.1836 & 234.8\\
			Mesh & Random & 20.42 & 6.92 & 0.1801 & 218.2\\
			Mesh & Credit & 20.42 & 6.92 & 0.2456 & 241.4\\
			Torus & Fixed & 16.72 & 5.00 & 0.1671 & 50.8\\
			Torus & Random & 16.72 & 5.02 & 0.1532 & 141.0\\
			Torus & Credit & 16.71 & 5.02 & 0.1744 & 64.0\\
			\bottomrule
		\end{tabular}
	\end{table}
	
	\begin{figure}[!htbp]
		\centering
		\resizebox{0.49\linewidth}{!}{%
			\begin{tikzpicture}
				\begin{axis}[hexaxis, hexlegendtop, xlabel={Injection rate},
					ylabel={Mean packet latency (cycles)}, ymode=log]
					\addplot+[hexline] table[x=rate,y=mesh_fixed] {\routinglat};
					\addlegendentry{Mesh fixed}
					\addplot+[hexline] table[x=rate,y=mesh_random] {\routinglat};
					\addlegendentry{Mesh random}
					\addplot+[hexline] table[x=rate,y=mesh_credit] {\routinglat};
					\addlegendentry{Mesh credit}
					\addplot+[hexline,dashed] table[x=rate,y=torus_fixed] {\routinglat};
					\addlegendentry{Torus fixed}
					\addplot+[hexline,dashed] table[x=rate,y=torus_random] {\routinglat};
					\addlegendentry{Torus random}
					\addplot+[hexline,dashed] table[x=rate,y=torus_credit] {\routinglat};
					\addlegendentry{Torus credit}
				\end{axis}
			\end{tikzpicture}
		}%
		\hfill
		\resizebox{0.49\linewidth}{!}{%
			\begin{tikzpicture}
				\begin{axis}[hexaxis, hexlegendtop, xlabel={Injection rate},
					ylabel={Accepted packet throughput}]
					\addplot+[hexline] table[x=rate,y=mesh_fixed] {\routingthr};
					\addlegendentry{Mesh fixed}
					\addplot+[hexline] table[x=rate,y=mesh_random] {\routingthr};
					\addlegendentry{Mesh random}
					\addplot+[hexline] table[x=rate,y=mesh_credit] {\routingthr};
					\addlegendentry{Mesh credit}
					\addplot+[hexline,dashed] table[x=rate,y=torus_fixed] {\routingthr};
					\addlegendentry{Torus fixed}
					\addplot+[hexline,dashed] table[x=rate,y=torus_random] {\routingthr};
					\addlegendentry{Torus random}
					\addplot+[hexline,dashed] table[x=rate,y=torus_credit] {\routingthr};
					\addlegendentry{Torus credit}
				\end{axis}
			\end{tikzpicture}
		}%
		\caption{Routing performance under all-node pair-list uniform traffic at $n=8$.}
		\label{fig:routing-performance}
	\end{figure}
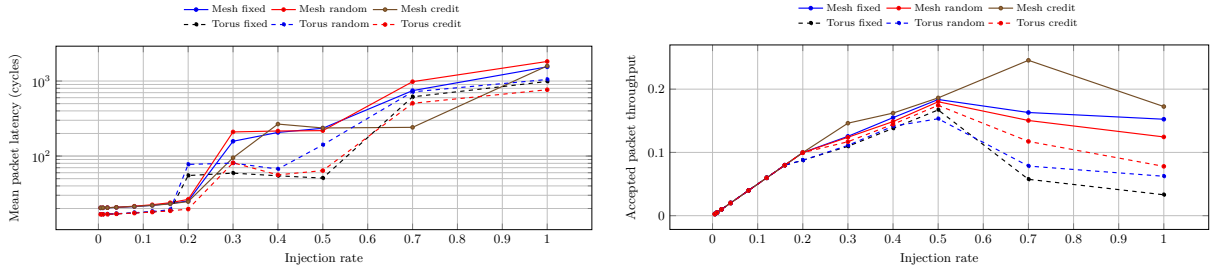
	\FloatBarrier
	
	\subsection{Comparison with conventional 2D mesh XY routing}
	
	To anchor the speedup claims against a conventional NoC baseline, we also
	evaluate gem5's standard $13\times13$ \texttt{Mesh\_XY} topology with
	deterministic XY routing. This baseline uses 169 routers, matching the
	$n=8$ HexMesh/HexTorus experiments, and reuses the same all-node
	uniform pair-list traffic files. The directory count remains 256 only to
	preserve the power-of-two destination encoding required by the pair-list
	Garnet standalone traffic generator; the active sources and destinations
	are the same 169 node ids used by the hex experiments.
	
	The conventional mesh has higher zero-load latency than both hex
	topologies. Its zero-load latency is $24.12$ cycles and its average hop
	count is $8.72$, compared with $20.42$ cycles and $6.92$ hops for
	HexMesh credit, and $16.71$ cycles and $5.02$ hops for HexTorus credit.
	Thus the hex mesh reduces zero-load latency by $1.18\times$ relative to
	the conventional 2D mesh, while the hex torus reduces it by
	$1.44\times$.
	
	For throughput, the strongest result is against the finite mesh:
	HexMesh credit reaches a maximum accepted packet throughput of $0.2456$,
	versus $0.1853$ for conventional Mesh\_XY, a $1.33\times$ improvement.
	HexTorus credit peaks at $0.1744$ under this traffic, slightly below the
	conventional mesh peak, although it keeps substantially lower latency at
	low and moderate injection rates. This reinforces the distinction between
	path-length benefit and high-load throughput: the torus shortens routes,
	but its wraparound paths can concentrate contention under this all-node
	pair-list workload. This comparison is matched by node count and traffic
	workload, rather than by router radix, physical link count, wiring budget,
	or area, and should therefore be interpreted as a topology- and
	routing-level comparison rather than a physical-design-normalized speedup.
	
	\begin{table}[tbp]
		\caption{Conventional $13\times13$ Mesh\_XY baseline compared with the credit-selected hex topologies.}
		\label{tab:meshxy-baseline}
		\centering
		\small
		\begin{tabular}{@{}lrrrrr@{}}
			\toprule
			Topology & ZLL & Avg. hops & Max thr. & Thr. vs. Mesh\_XY & ZLL reduction\\
			\midrule
			Mesh\_XY & 24.12 & 8.72 & 0.1853 & $1.00\times$ & $1.00\times$\\
			HexMesh credit & 20.42 & 6.92 & 0.2456 & $1.33\times$ & $1.18\times$\\
			HexTorus credit & 16.71 & 5.02 & 0.1744 & $0.94\times$ & $1.44\times$\\
			\bottomrule
		\end{tabular}
	\end{table}
	
	\begin{figure}[!htbp]
		\centering
		\resizebox{0.49\linewidth}{!}{%
			\begin{tikzpicture}
				\begin{axis}[hexaxis, hexlegendtop, xlabel={Injection rate},
					ylabel={Mean packet latency (cycles)}, ymode=log]
					\addplot+[hexline,densely dotted] table[x=rate,y=meshxy_lat] {\baselineperf};
					\addlegendentry{Mesh\_XY}
					\addplot+[hexline] table[x=rate,y=hexmesh_credit_lat] {\baselineperf};
					\addlegendentry{HexMesh credit}
					\addplot+[hexline,dashed] table[x=rate,y=hextorus_credit_lat] {\baselineperf};
					\addlegendentry{HexTorus credit}
				\end{axis}
			\end{tikzpicture}
		}%
		\hfill
		\resizebox{0.49\linewidth}{!}{%
			\begin{tikzpicture}
				\begin{axis}[hexaxis, hexlegendtop, xlabel={Injection rate},
					ylabel={Accepted packet throughput}]
					\addplot+[hexline,densely dotted] table[x=rate,y=meshxy_thr] {\baselineperf};
					\addlegendentry{Mesh\_XY}
					\addplot+[hexline] table[x=rate,y=hexmesh_credit_thr] {\baselineperf};
					\addlegendentry{HexMesh credit}
					\addplot+[hexline,dashed] table[x=rate,y=hextorus_credit_thr] {\baselineperf};
					\addlegendentry{HexTorus credit}
				\end{axis}
			\end{tikzpicture}
		}%
		\caption{Comparison against a conventional $13\times13$ Mesh\_XY baseline under the same all-node pair-list traffic.}
		\label{fig:meshxy-baseline}
	\end{figure}
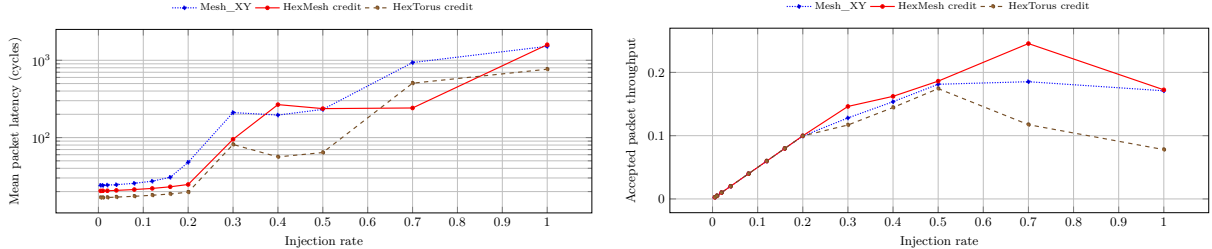
	\FloatBarrier
	
	\subsection{Pressure in the restricted sectors}
	
	The proofs in Sections~\ref{sec:mesh-routing} and
	\ref{sec:torus-routing} show that the two boundary sectors affected by
	the forbidden turns retain only the unique group-monotone direction word,
	whereas the internal sectors retain all minimal interleavings. To isolate
	this structural cost, the sector experiment compares boundary-sector and
	internal-sector traffic while matching the distance distribution.
	Figure~\ref{fig:sector-pressure} reports the boundary/internal latency
	ratio for representative high-load rates. Ratios above one indicate that
	boundary-sector traffic has higher latency than internal-sector traffic.
	The ratio is often above one under high load, especially for torus fixed
	and credit and for mesh credit. The trend is not perfectly monotone and
	random selection has more variability, but the experiment supports the
	structural expectation that the cost of the two-turn restriction is
	concentrated in the sectors where path diversity is intentionally
	reduced.
	Table~\ref{tab:sector-summary} summarizes the throughput side of the
	same experiment. For fixed and credit selection, internal-sector traffic
	reaches higher maximum throughput than boundary-sector traffic on both
	topologies. The random rows are less regular, which is consistent with
	randomized candidate choice adding run-to-run variation on top of the
	sector structure.
	
	\begin{table}[H]
		\caption{Sector-pressure summary at $n=8$. Boundary and internal traffic use matched distance distributions.}
		\label{tab:sector-summary}
		\centering
		\small
		\begin{tabular}{@{}llrrrr@{}}
			\toprule
			Topology & Selection & Boundary ZLL & Internal ZLL & Boundary max thr. & Internal max thr.\\
			\midrule
			Mesh & Fixed & 15.72 & 15.73 & 0.2282 & 0.2362\\
			Mesh & Random & 15.57 & 15.61 & 0.2276 & 0.2199\\
			Mesh & Credit & 15.57 & 15.61 & 0.2276 & 0.2568\\
			Torus & Fixed & 15.62 & 15.53 & 0.1808 & 0.2109\\
			Torus & Random & 15.60 & 15.71 & 0.1800 & 0.1639\\
			Torus & Credit & 15.60 & 15.71 & 0.1800 & 0.2042\\
			\bottomrule
		\end{tabular}
	\end{table}
	
	\begin{figure}[H]
		\centering
		\begin{tikzpicture}
			\begin{axis}[hexaxis, width=0.82\linewidth, height=0.44\linewidth,
				xlabel={Injection rate}, ylabel={Boundary/Internal latency ratio},
				legend pos=north west, ymin=0.5]
				\addplot+[hexline] table[x=rate,y=mesh_fixed] {\sectorratio};
				\addlegendentry{Mesh fixed}
				\addplot+[hexline] table[x=rate,y=mesh_credit] {\sectorratio};
				\addlegendentry{Mesh credit}
				\addplot+[hexline,dashed] table[x=rate,y=torus_fixed] {\sectorratio};
				\addlegendentry{Torus fixed}
				\addplot+[hexline,dashed] table[x=rate,y=torus_credit] {\sectorratio};
				\addlegendentry{Torus credit}
				\addplot+[black, dotted, domain=0.08:0.7] {1};
			\end{axis}
		\end{tikzpicture}
		\caption{Boundary-sector pressure relative to internal-sector traffic.}
		\label{fig:sector-pressure}
	\end{figure}
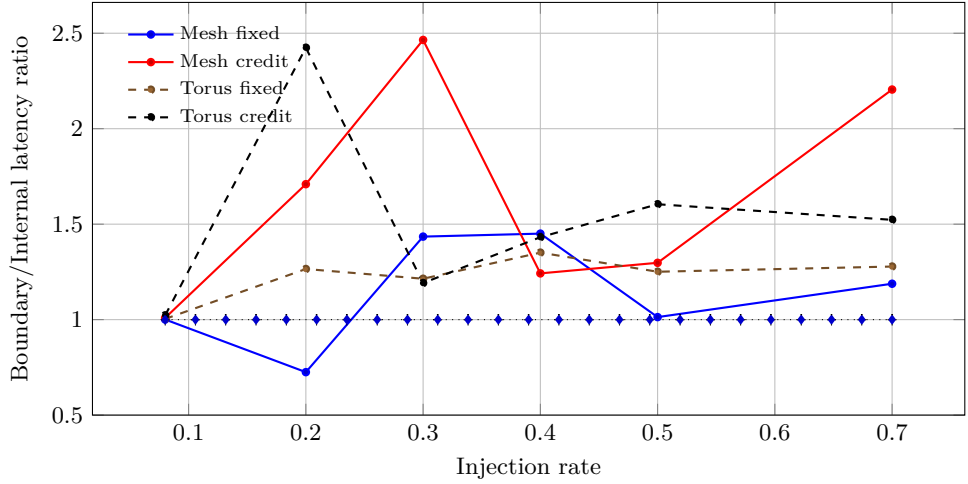
	\FloatBarrier
	\vspace{0.5\baselineskip}
	
	\subsection{Torus dateline VC behavior}
	
	Figure~\ref{fig:vc-behavior} evaluates whether the Hamiltonian dateline
	state machine is exercised in simulation. Under all-node uniform traffic,
	VC1 accounts for about $0.28$ of torus channel hops with credit
	selection. Under dateline-heavy traffic, the VC1 hop fraction rises to
	about $0.76$--$0.78$. Dateline VC stalls also increase by more than an
	order of magnitude. These counters show that the torus experiments are
	not merely using two VCs as a static buffer pool; they are exercising the
	per-hop phase transition specified in~\eqref{eq:vc-map}. All runs in this
	set complete without deadlock.
	Table~\ref{tab:vc-summary} gives both the traffic-level summary and
	representative per-rate VC counters for credit selection. Dateline-heavy
	traffic has about $3.1$ more zero-load cycles than uniform traffic and a
	higher average hop count, reflecting the longer wraparound-oriented
	paths used to stress the dateline rule. At rate $0.4$, dateline-heavy
	traffic has roughly $3.4\times$ the VC1 hop fraction and about
	$34\times$ the dateline VC stalls of uniform traffic.
	The VC-behavior experiment uses a restricted injection-rate range, so its
	maximum observed throughput is not the global saturation throughput
	reported in Section~\ref{sec:routing-performance}.
	
	\begin{table}[tbp]
		\centering
		\small
		\caption{Torus VC behavior at $n=8$. The first block summarizes traffic-level performance; the second block reports credit-selection VC counters at selected rates.}
		\label{tab:vc-summary}
		\begin{tabular}{@{}llrrrr@{}}
			\toprule
			Traffic & Selection & ZLL & Avg. hops & Max obs. thr. & Lat. at max thr.\\
			\midrule
			Uniform & Fixed & 16.72 & 5.00 & 0.1381 & 54.6\\
			Uniform & Random & 16.72 & 5.02 & 0.1414 & 67.1\\
			Uniform & Credit & 16.71 & 5.02 & 0.1445 & 56.2\\
			Dateline-heavy & Fixed & 19.81 & 6.55 & 0.1464 & 525.9\\
			Dateline-heavy & Random & 19.80 & 6.53 & 0.1409 & 641.0\\
			Dateline-heavy & Credit & 19.78 & 6.54 & 0.1534 & 455.2\\
			\midrule
			Traffic & Rate & Latency & Throughput & VC1 frac. & DL VC stalls\\
			\midrule
			Uniform & 0.04 & 16.97 & 0.0200 & 0.278 & 792\\
			Uniform & 0.20 & 19.63 & 0.0996 & 0.279 & 30,432\\
			Uniform & 0.40 & 56.18 & 0.1445 & 0.288 & 42,531\\
			Dateline-heavy & 0.04 & 20.38 & 0.0200 & 0.762 & 4,309\\
			Dateline-heavy & 0.20 & 145.23 & 0.0952 & 0.768 & 539,903\\
			Dateline-heavy & 0.40 & 455.20 & 0.1534 & 0.777 & 1,437,240\\
			\bottomrule
		\end{tabular}
	\end{table}
	
	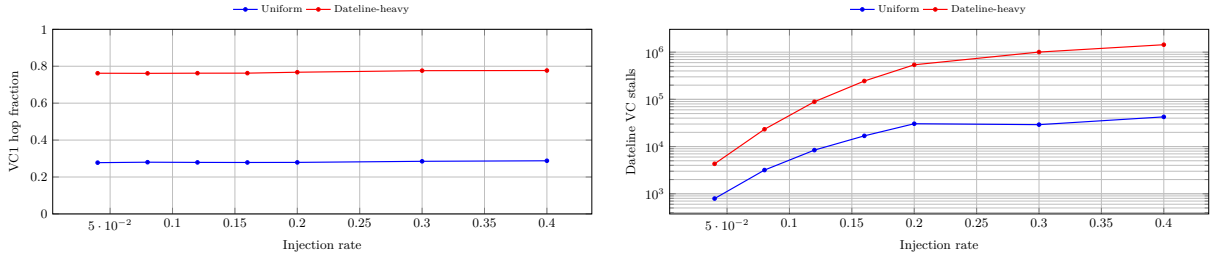
\begin{figure}[!htbp]
		\centering
		\resizebox{0.49\linewidth}{!}{%
			\begin{tikzpicture}
				\begin{axis}[hexaxis, xlabel={Injection rate}, ylabel={VC1 hop fraction},
					ymin=0, ymax=1,
					legend style={at={(0.5,1.05)}, anchor=south, legend columns=2,
						font=\scriptsize, cells={anchor=west}, draw=none, fill=none}]
					\addplot+[hexline] table[x=rate,y=uniform_vc1] {\vcdata};
					\addlegendentry{Uniform}
					\addplot+[hexline] table[x=rate,y=dateline_vc1] {\vcdata};
					\addlegendentry{Dateline-heavy}
				\end{axis}
			\end{tikzpicture}
		}%
		\hfill
		\resizebox{0.49\linewidth}{!}{%
			\begin{tikzpicture}
				\begin{axis}[hexaxis, xlabel={Injection rate}, ylabel={Dateline VC stalls},
					ymode=log,
					legend style={at={(0.5,1.05)}, anchor=south, legend columns=2,
						font=\scriptsize, cells={anchor=west}, draw=none, fill=none}]
					\addplot+[hexline] table[x=rate,y=uniform_dlstall] {\vcdata};
					\addlegendentry{Uniform}
					\addplot+[hexline] table[x=rate,y=dateline_dlstall] {\vcdata};
					\addlegendentry{Dateline-heavy}
				\end{axis}
			\end{tikzpicture}
		}%
		\caption{Torus dateline VC behavior with credit selection.}
		\label{fig:vc-behavior}
	\end{figure}
	\FloatBarrier
	
	\subsection{Scaling}
	
	Figure~\ref{fig:scaling} reports scaling from $n=4$ to $n=12$ with 13
	injection rates and three seeds per point. Zero-load latency increases
	with $n$, as expected from the growing average minimal distance. Mesh
	latency grows from about $13.0$ cycles at $n=4$ to $20.4$ cycles at
	$n=8$ and $27.6$--$27.7$ cycles at $n=12$. Torus latency grows from
	about $11.25$ cycles to $16.7$ cycles and then to about $22.1$ cycles.
	Thus the torus keeps a lower zero-load latency at every evaluated size.
	
	Maximum accepted throughput decreases as the topology grows. Mesh fixed
	peaks at $0.3836$, $0.1836$, and $0.1417$ for $n=4,8,12$, while mesh
	credit peaks at $0.4025$, $0.2456$, and $0.1477$. Torus fixed peaks at
	$0.4002$, $0.1671$, and $0.1019$, and torus credit peaks at $0.4012$,
	$0.1744$, and $0.1040$. The results support the expected geometry-driven
	latency trend and show that credit selection helps mesh throughput more
	than torus throughput under this traffic.
	Table~\ref{tab:scaling-summary} gives the same trend in tabular form.
	Each row contains 39 samples: 13 injection rates and three seeds. The
	$n=12$ results therefore use the same statistical structure as the
	smaller sizes rather than a one-off smoke test.
	
	\begin{table}[tbp]
		\caption{Scaling summary. Each row aggregates 13 injection rates and three seeds.}
		\label{tab:scaling-summary}
		\centering
		\small
		\begin{tabular}{@{}llrrrr@{}}
			\toprule
			Topology & Selection & $n$ & ZLL & Avg. hops & Max thr.\\
			\midrule
			Mesh & Fixed & 4 & 13.02 & 3.22 & 0.3836\\
			Mesh & Fixed & 8 & 20.36 & 6.81 & 0.1836\\
			Mesh & Fixed & 12 & 27.70 & 10.52 & 0.1417\\
			Mesh & Credit & 4 & 12.98 & 3.21 & 0.4025\\
			Mesh & Credit & 8 & 20.42 & 6.92 & 0.2456\\
			Mesh & Credit & 12 & 27.63 & 10.44 & 0.1477\\
			Torus & Fixed & 4 & 11.26 & 2.36 & 0.4002\\
			Torus & Fixed & 8 & 16.72 & 5.00 & 0.1671\\
			Torus & Fixed & 12 & 22.07 & 7.66 & 0.1019\\
			Torus & Credit & 4 & 11.25 & 2.32 & 0.4012\\
			Torus & Credit & 8 & 16.71 & 5.02 & 0.1744\\
			Torus & Credit & 12 & 22.07 & 7.67 & 0.1040\\
			\bottomrule
		\end{tabular}
	\end{table}
	
	\begin{figure}[!htbp]
		\centering
		\resizebox{0.49\linewidth}{!}{%
			\begin{tikzpicture}
				\begin{axis}[hexaxis, xlabel={$n$}, ylabel={Zero-load latency (cycles)},
					xtick={4,8,12},
					legend style={at={(0.5,1.05)}, anchor=south, legend columns=4,
						font=\scriptsize, cells={anchor=west}, draw=none, fill=none}]
					\addplot+[hexline] table[x=n,y=mesh_fixed_zll] {\scalingdata};
					\addlegendentry{Mesh fixed}
					\addplot+[hexline] table[x=n,y=mesh_credit_zll] {\scalingdata};
					\addlegendentry{Mesh credit}
					\addplot+[hexline,dashed] table[x=n,y=torus_fixed_zll] {\scalingdata};
					\addlegendentry{Torus fixed}
					\addplot+[hexline,dashed] table[x=n,y=torus_credit_zll] {\scalingdata};
					\addlegendentry{Torus credit}
				\end{axis}
			\end{tikzpicture}
		}%
		\hfill
		\resizebox{0.49\linewidth}{!}{%
			\begin{tikzpicture}
				\begin{axis}[hexaxis, xlabel={$n$}, ylabel={Max accepted throughput},
					xtick={4,8,12},
					legend style={at={(0.5,1.05)}, anchor=south, legend columns=4,
						font=\scriptsize, cells={anchor=west}, draw=none, fill=none}]
					\addplot+[hexline] table[x=n,y=mesh_fixed_thr] {\scalingdata};
					\addlegendentry{Mesh fixed}
					\addplot+[hexline] table[x=n,y=mesh_credit_thr] {\scalingdata};
					\addlegendentry{Mesh credit}
					\addplot+[hexline,dashed] table[x=n,y=torus_fixed_thr] {\scalingdata};
					\addlegendentry{Torus fixed}
					\addplot+[hexline,dashed] table[x=n,y=torus_credit_thr] {\scalingdata};
					\addlegendentry{Torus credit}
				\end{axis}
			\end{tikzpicture}
		}%
		\caption{Scaling result for $n=4,8,12$.}
		\label{fig:scaling}
	\end{figure}
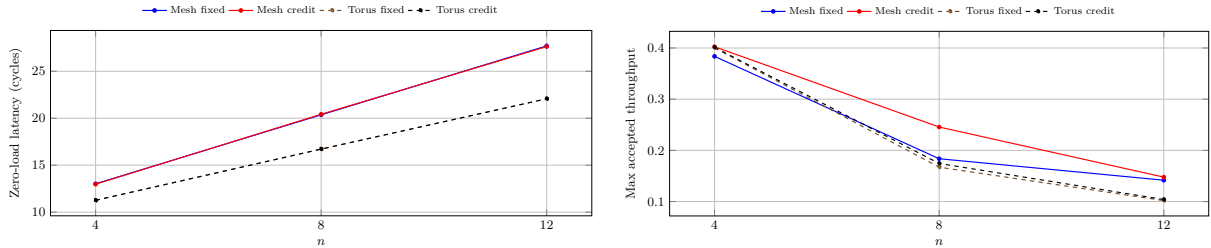
	\FloatBarrier
	
	\section{Discussion and Limitations}
	
	The two algorithms share the same geometric turn restriction but use
	different global ordering mechanisms. In the finite mesh, the linear
	potential $\Phi(x,y)=x+2y$ is well-defined and strictly monotone within
	each direction group. This potential does not descend to a well-defined
	strict order on the periodic quotient: the torus contains directed
	non-contractible channel cycles, and no single-valued function on
	quotient nodes can increase strictly around a directed cycle.
	Proposition~\ref{prop:torus-one-vc-cycle} makes this obstruction explicit
	for the one-VC version of our torus physical relation. The Hamiltonian
	coordinates instead provide cyclic orders, and the two VC phases cut
	those orders into an acyclic resource ranking.
	
	The main minimal-routing guarantees rely on wormhole routing with reliable
	bidirectional links, no U-turns, and the specific period lattice
	in~\eqref{eq:periods}. Proposition~\ref{prop:single-link-fault} establishes
	a separate exception for one static, globally known bidirectional link
	failure: only the prescribed same-group triangle substitutions are added,
	and the resulting routes are bounded relative to the healthy shortest-path
	distance rather than claimed to remain minimal. This result does not cover
	multiple link failures, node failures, arbitrary nonminimal detours,
	multicast dependencies, or dynamic topology changes. In particular, it
	does not establish acyclicity for the union of old and new routing
	configurations during an in-flight reconfiguration; such settings require
	a quiescent transition, a new CDG proof, or a separate deadlock-free escape
	subnetwork such as those studied in more general adaptive-routing
	theory~\citep{duato1993}. The torus router must also support per-hop VC allocation and
	the reset $\VCone\rightarrow\VCzero$ when the direction group increases;
	an implementation that fixes one VC for an entire packet does not realize
	the proposed algorithm.
	
	The evaluation is intentionally focused on the routing relation and its
	VC mechanism rather than on full-system application behavior. We use
	Garnet synthetic traffic and do not report application traces, power,
	area, timing closure, or router critical-path measurements. The evaluation
	covers all-node uniform routing performance, a conventional 2D-mesh
	baseline, boundary-sector pressure, dateline-heavy torus traffic, and
	scaling to $n=12$, but it does not cover hotspot traffic or
	application-driven burstiness. The torus results
	should also be interpreted carefully: lower zero-load latency follows
	from shorter average paths, but high-load throughput depends on how
	traffic maps onto the quotient links and can be lower than the finite
	mesh under the tested all-node pair-list workload. The single-link fault
	extension in Proposition~\ref{prop:single-link-fault} is a theoretical
	result; this paper does not add fault-injection performance experiments.
	
	A further theoretical question is whether a different construction can
	preserve comparable minimal adaptivity with fewer turn restrictions or
	fewer VC resources. Proposition~\ref{prop:torus-one-vc-cycle} shows that
	the single-VC resource labelling of our particular torus physical
	relation $\RT$ has a cyclic complete CDG for $n\geq3$, while
	Theorem~\ref{thm:torus} establishes that two VCs are sufficient for the
	Hamiltonian-dateline construction. We do not claim that two VCs are
	necessary for every minimal routing relation or under alternative
	flow-control mechanisms. The simulator results confirm that the proposed
	two-VC rule is implementable and exercised by dateline-heavy traffic,
	but they do not prove such a general necessity result.
	
	\section{Conclusion}
	
	We presented a unified coordinate model and two deadlock-free minimal
	routing constructions for degree-six triangular-lattice networks. The
	finite mesh uses one VC and two directed turn prohibitions; its complete
	resource CDG is ordered by a linear potential. For the periodic torus,
	every source--destination pair has a unique closest lift, yet assigning a single VC to the same
	two-turn physical relation yields a cyclic resource CDG for every
	$n\geq3$. Two VC phases and Hamiltonian datelines cut this residual
	periodic cycle into a strict global resource rank. Both physical routing
	relations preserve minimal all-pairs connectivity. The mesh relation is
	partially adaptive for every $n\geq2$, while the torus relation is
	partially adaptive exactly when $n\geq3$.
	We implemented the routing relations in gem5/Garnet, verified them with
	an independent CDG checker, and evaluated routing performance, restricted
	sector pressure, dateline VC behavior, and scaling. The measurements
	support the main design expectations: the torus has lower zero-load
	latency, congestion-aware selection can exploit the retained mesh
	adaptivity under load, boundary-sector traffic exposes the cost of the
	two forbidden turns, and dateline-heavy traffic exercises the torus VC
	state machine. The same geometric framework also admits the restricted
	static single-link extension of Proposition~\ref{prop:single-link-fault}:
	a rotated turn configuration and local triangle bypass preserve
	connectivity and deadlock freedom with at most one additional hop. General
	multiple-fault and dynamic-fault routing remain open.
	
	\subsection*{AI use statement}
	Generative AI tools were used to assist with language editing, LaTeX organization, literature discovery, preliminary consistency checks, and experiment-script development. The authors are responsible for verifying every citation, mathematical statement, proof, implementation, and experimental result in this manuscript. No experimental measurements or empirical claims in this draft were generated or fabricated by an AI system.
	
	\subsection*{Ethics statement}
	This work studies network routing algorithms and does not involve human participants, personal data, or deployment of a decision-making system. We are not aware of direct ethical risks beyond the general need to report proofs, implementation details, and experimental results accurately.
	
	\subsection*{Reproducibility statement}
	All topology definitions, routing rules, VC transitions, and theorem assumptions are stated in Sections~\ref{sec:model}--\ref{sec:routing}. Complete proofs are included in the appendix. The implementation includes an independent CDG checker, traffic-pair generators, experiment runners, result checkers, and summary scripts for the simulation configuration described in Section~\ref{sec:evaluation}.
	
	\bibliographystyle{iclr2027_conference}
	\bibliography{iclr2027_conference}
	
	\clearpage
	\appendix
	\section{Basic Triangular-Lattice Facts}
\label{app:latticefacts}

\begin{proof}[Proof of Lemma~\ref{lem:decomposition}]
	The six closed cones generated by adjacent directions cover
	$\mathbb Z^2$; neighboring cones meet only along their boundary rays,
	and their interiors are pairwise disjoint. In each cone, the required
	coefficients are
	\begin{equation}
		\begin{array}{c|c|c|c}
			\text{sector} & \text{conditions} & (a,b) & a+b\\ \hline
			(\dir{0},\dir{1}) & x\geq0,\ y\geq0
			& (x,y) & x+y\\
			(\dir{1},\dir{2}) & x\leq0,\ x+y\geq0
			& (x+y,-x) & y\\
			(\dir{2},\dir{3}) & y\geq0,\ x+y\leq0
			& (y,-x-y) & -x\\
			(\dir{3},\dir{4}) & x\leq0,\ y\leq0
			& (-x,-y) & -x-y\\
			(\dir{4},\dir{5}) & x\geq0,\ x+y\leq0
			& (-x-y,x) & -y\\
			(\dir{5},\dir{0}) & y\leq0,\ x+y\geq0
			& (-y,x+y) & x
		\end{array}
		\label{eq:sector-decomposition}
	\end{equation}
	where the ordered pair $(a,b)$ multiplies the two directions in the
	first column. The conditions in every row make both coefficients
	nonnegative. Direct substitution gives
	$(x,y)=a\dir{i}+b\dir{i+1}$, and the final column equals
	$\max\{|x|,|y|,|x+y|\}$ under the corresponding sector conditions.
	Boundary points may belong to two adjacent cones, but either
	representation satisfies the claim.
\end{proof}

\begin{proof}[Proof of Lemma~\ref{lem:lattice-distance}]
	Every direction in~\eqref{eq:directions} has hexagonal norm one. The
	triangle inequality therefore implies that every lattice path from $s$
	to $t$ has length at least $\hexnorm{t-s}$. By
	Lemma~\ref{lem:decomposition}, the displacement $t-s$ has a
	representation by $a+b=\hexnorm{t-s}$ unit steps. This representation
	gives a path attaining the lower bound, so the distance equals
	$\hexnorm{t-s}$.
\end{proof}

\begin{proof}[Proof of Lemma~\ref{lem:geodesic-directions}]
	Suppose a path contains two nonadjacent directions. Up to rotation and
	reversal of the cyclic direction order, the two directions have circular
	separation two or three. For every $i$ (indices modulo six),
	\begin{equation}
		\dir{i}+\dir{i+2}=\dir{i+1},
		\qquad
		\dir{i}+\dir{i+3}=0.
		\label{eq:direction-shortening}
	\end{equation}
	Thus two steps whose directions have separation two can be replaced by
	one unit step with the same net displacement, while two opposite steps
	can be deleted. Separation four is the rotated form of separation two.
	Since a path endpoint depends only on the sum of its step vectors, either
	replacement produces a strictly shorter lattice walk with the same
	endpoints. A geodesic therefore cannot contain a nonadjacent pair of
	directions. Any pairwise-adjacent subset of the six-cycle of directions
	has size at most two, which proves the claim.
\end{proof}

\section{Proof of the Mesh Routing Theorem}
\label{app:meshproof}

\begin{proof}[Proof of Theorem~\ref{thm:mesh}]
	\medskip\noindent\textbf{Minimal all-pairs connectivity.}
	Let $s,t\in V_n$ and $\Delta=t-s$. By
	Lemma~\ref{lem:decomposition}, write
	\begin{equation}
		\Delta=a\dir{i}+b\dir{i+1},
		\qquad a,b\geq0,
		\qquad a+b=\hexnorm{\Delta}.
		\label{eq:mesh-proof-decomposition}
	\end{equation}
	If the sector is neither $\{\dir{5},\dir{0}\}$ nor
	$\{\dir{2},\dir{3}\}$, every interleaving of the $a+b$ steps avoids
	both turns in~\eqref{eq:bannedturns}. In sector
	$\{\dir{5},\dir{0}\}$, the word $\dir{5}^{a}\dir{0}^{b}$ avoids
	$\dir{0}\rightarrow\dir{5}$. In sector
	$\{\dir{2},\dir{3}\}$, the word $\dir{3}^{b}\dir{2}^{a}$ avoids
	$\dir{2}\rightarrow\dir{3}$. Every such word has lattice length $\hexnorm{\Delta}$ by
	Lemma~\ref{lem:lattice-distance}. Since every mesh path is also a lattice
	path, no path inside the mesh can be shorter.
	
	It remains to verify that the constructed path stays in the finite mesh.
	The vertex set is the intersection of the three strips
	\begin{equation}
		-R\leq x\leq R,
		\qquad -R\leq y\leq R,
		\qquad -R\leq x+y\leq R.
		\label{eq:mesh-strips}
	\end{equation}
	Along a path formed from two adjacent directions, each of $x$, $y$, and
	$x+y$ is monotone or constant. Hence every intermediate value lies
	between its values at $s$ and $t$, both of which satisfy
	\eqref{eq:mesh-strips}. Every intermediate node therefore belongs to
	$V_n$. Thus $\RM(s,t)\neq\varnothing$ for every ordered pair, and every
	path in $\RM$ is minimal by definition.
	
	\medskip\noindent\textbf{Retained adaptivity.}
	Suppose first that
	$\Delta=a\dir{i}+b\dir{i+1}\in S_i^\circ$, so $a,b>0$. In a
	within-group sector, every ordering of $a$ copies of one direction and
	$b$ copies of the other is permitted, giving $\binom{a+b}{a}$ direction
	words. In a cross-group sector, a word that places an upper-group step
	before a later lower-group step must contain the corresponding adjacent
	transition $\dir{0}\rightarrow\dir{5}$ or
	$\dir{2}\rightarrow\dir{3}$ somewhere in the word. Hence only the
	$\mathrm L$-before-$\mathrm U$ word is permitted. If $a=0$ or $b=0$,
	the displacement lies on a sector boundary ray and its minimal direction
	word is unique. For $n\geq2$, the displacement from $(-1,0)$ to $(0,1)$
	has the two permitted shortest direction words $\dir{0}\dir{1}$ and
	$\dir{1}\dir{0}$. Both first channels are legal at their common source,
	so $\RM$ is partially adaptive according to the definition in
	Section~\ref{sec:model}.
	
	\medskip\noindent\textbf{Deadlock freedom.}
	Consider any path in $\RM$. Because its length is
	$d_{\mathrm H}(s,t)$, it is also a geodesic of the infinite triangular
	lattice. Lemma~\ref{lem:geodesic-directions} therefore confines its
	direction word to one adjacent-direction sector. A within-group sector
	lies wholly in group $\mathrm L$ or wholly in group $\mathrm U$. In the
	two cross-group sectors, the prohibited turns force the lower-group
	directions to precede the upper-group directions. Consequently every
	permitted path has group sequence $\mathrm L^*\mathrm U^*$, and every
	edge of the complete one-VC CDG is either within one group or directed
	from $\mathrm L$ to $\mathrm U$. A CDG cycle containing both groups would
	require an edge from $\mathrm U$ back to $\mathrm L$, which does not
	exist.
	
	It remains to exclude a cycle contained in a single group. Define
	$\Phi(x,y)=x+2y$. Its direction increments are
	\begin{equation}
		\begin{aligned}
			\Delta\Phi(\dir{0})&=1,&
			\Delta\Phi(\dir{1})&=2,&
			\Delta\Phi(\dir{2})&=1,\\
			\Delta\Phi(\dir{3})&=-1,&
			\Delta\Phi(\dir{4})&=-2,&
			\Delta\Phi(\dir{5})&=-1.
		\end{aligned}
		\label{eq:mesh-potential-increments}
	\end{equation}
	For a dependency $c\rightarrow c'$ within group $\mathrm U$, the tail
	of $c'$ is the head of $c$, so its potential is strictly larger than the
	tail potential of $c$. Along a dependency chain wholly in $\mathrm U$,
	these tail potentials therefore strictly increase. They strictly
	decrease along every chain wholly in $\mathrm L$. Neither chain can
	close. Hence the complete one-VC resource CDG of $\RM$ is acyclic, and
	$\RM$ is deadlock-free.
\end{proof}

\section{Geometry of the Periodic Torus}
\label{app:torusgeometry}

Let $R_{60}(x,y)=(-y,x+y)$. This map cyclically permutes the directions,
preserves the hexagonal norm, and satisfies
$R_{60}(T_1)=T_2$ and $R_{60}(T_2)=T_2-T_1$. Consequently,
$\lattice$ is invariant under rotations by $60^\circ$.

\begin{lemma}[Minimum period length]
	\label{lem:minperiod}
	Every nonzero $\lambda\in\lattice$ satisfies
	$\hexnorm{\lambda}\geq2n-1$.
\end{lemma}

\begin{proof}
	Rotate $\lambda$ into the sector $x\geq0$, $y\geq0$, where
	$\hexnorm{(x,y)}=x+y=:s$. The rotation leaves $\lattice$ invariant.
	Direct calculation from~\eqref{eq:periods} shows that $x+ky$ is an
	integer multiple of $N$ for every $(x,y)\in\lattice$. For a nonzero
	vector in this sector, $x+ky>0$. Suppose for contradiction that
	$s\leq2n-2$. Since $0\leq y\leq s$,
	\begin{equation}
		0<x+ky=s+(3n-2)y
		\leq(3n-1)s
		\leq(3n-1)(2n-2)<2N.
		\label{eq:minperiod-bound}
	\end{equation}
	Hence $x+ky=N$ and
	\begin{equation}
		s=N-(3n-2)y.
		\label{eq:minperiod-s}
	\end{equation}
	If $y\leq n-2$, then
	$s\geq N-(3n-2)(n-2)=5n-3>2n-2$.
	If $y=n-1$, then $s=2n-1$.
	If $y\geq n$, then
	$s\leq N-n(3n-2)=1-n<0$.
	Every case contradicts $0\leq y\leq s\leq2n-2$. Therefore
	$s\geq2n-1$.
\end{proof}

\begin{proof}[Proof of Lemma~\ref{lem:representatives}]
	If two distinct points of $V_n$ represented the same quotient node,
	their difference would be a nonzero period vector. By the triangle
	inequality, that difference would have norm at most $2(n-1)$, contrary
	to Lemma~\ref{lem:minperiod}. Thus the points in $V_n$ represent
	distinct quotient nodes. By~\eqref{eq:nodecount},
	$|V_n|=3n^2-3n+1=N$, while~\eqref{eq:N} shows that the quotient also
	has $N$ nodes. Hence $V_n$ is a complete, nonredundant representative
	set.
\end{proof}

\begin{proof}[Proof of Corollary~\ref{cor:diameter}]
	Lemma~\ref{lem:representatives} implies that every quotient displacement
	has a representative in $V_n$, whose norm is at most $n-1$. Therefore
	$\operatorname{diam}(\torus)\leq n-1$.
	
	For the reverse inequality, let $v=(n-1,0)$, whose norm is $n-1$. For
	every nonzero $\lambda\in\lattice$, the reverse triangle inequality and
	Lemma~\ref{lem:minperiod} give
	\begin{equation}
		\hexnorm{v+\lambda}
		\geq\hexnorm{\lambda}-\hexnorm{v}
		\geq(2n-1)-(n-1)=n.
		\label{eq:diameter-lower-bound}
	\end{equation}
	Thus $v$ itself is the unique closest representative of its quotient
	class, and its distance from $[0]$ is $n-1$. Hence
	$\operatorname{diam}(\torus)\geq n-1$, proving equality.
\end{proof}

\begin{proof}[Proof of Corollary~\ref{cor:unique-lift}]
	Existence follows because the period lattice is discrete and the
	hexagonal norm is proper on $\mathbb Z^2$, so the minimum in
	\eqref{eq:torusdistance} is attained. Suppose two distinct
	$\lambda_1,\lambda_2\in\lattice$ both attain it. Let
	\[
	\Delta_j=t-s+\lambda_j,\qquad j\in\{1,2\}.
	\]
	By Corollary~\ref{cor:diameter},
	$\hexnorm{\Delta_1}=\hexnorm{\Delta_2}
	=d_{\torus}([s],[t])\leq n-1$. Hence
	\begin{equation}
		\hexnorm{\lambda_1-\lambda_2}
		=\hexnorm{\Delta_1-\Delta_2}
		\leq \hexnorm{\Delta_1}+\hexnorm{\Delta_2}
		\leq 2n-2.
		\label{eq:unique-lift-contradiction}
	\end{equation}
	But $\lambda_1-\lambda_2$ is a nonzero period vector, contradicting
	Lemma~\ref{lem:minperiod}, which gives norm at least $2n-1$. Therefore
	the closest lift is unique.
\end{proof}

\begin{proof}[Proof of Lemma~\ref{lem:hamiltonian}]
	Direct calculation gives
	\begin{equation}
		\HU(T_1)=n+k(n-1)=N,
		\qquad
		\HU(T_2)=-(n-1)+k(2n-1)=2N.
		\label{eq:ham-periods}
	\end{equation}
	Thus $\HU$ is invariant modulo $N$ under every period translate and is
	well-defined on $\torus$. For every residue $r\in\mathbb Z_N$, the
	quotient node $[r\dir{0}]$ satisfies $\HU([r\dir{0}])=r$ because
	$\HU(\dir{0})=1$. Hence $\HU$ is surjective. The domain $\torus$ and
	codomain $\mathbb Z_N$ both contain $N$ elements, so $\HU$ is a
	bijection. Since $\HL=-\HU\pmod N$, $\HL$ is also a bijection. Finally,
	$\dir{0}$ increments $\HU$ by one and $\dir{3}$ increments $\HL$ by
	one; repeated channels in either direction therefore visit all quotient
	nodes in the corresponding cyclic Hamiltonian order.
\end{proof}

\section{Proof of the Torus Routing Theorem}
\label{app:torusproof}

\begin{proof}[Proof of Proposition~\ref{prop:torus-one-vc-cycle}]
	For $r\in\mathbb Z_N$, let
	\[
	c_r:[r\dir{0}]\longrightarrow[(r+1)\dir{0}].
	\]
	By Lemma~\ref{lem:hamiltonian}, the nodes $[r\dir{0}]$ for
	$r\in\mathbb Z_N$ are all distinct and occur in the cyclic order of
	$\HU$, so $c_0,\ldots,c_{N-1}$ form a directed physical Hamiltonian
	cycle.
	
	For each $r$, consider the ordered quotient pair
	$([r\dir{0}],[(r+2)\dir{0}])$. Let $s_r,t_r\in V_n$ be its canonical
	representatives. Since the quotient displacement is represented by
	$2\dir{0}$, there is a period vector $\lambda_r$ such that
	$t_r-s_r+\lambda_r=2\dir{0}$. For every nonzero
	$\mu\in\lattice$, Lemma~\ref{lem:minperiod} and the reverse triangle
	inequality give, for $n\geq3$,
	\begin{equation}
		\hexnorm{2\dir{0}+\mu}
		\geq \hexnorm{\mu}-\hexnorm{2\dir{0}}
		\geq (2n-1)-2
		\geq 3.
		\label{eq:one-vc-twohop-minimal}
	\end{equation}
	Thus $2\dir{0}$ is the unique closest lifted displacement for this
	ordered pair. The two-hop direction word $\dir{0}\dir{0}$ is therefore
	a retained route in $\RT$. It contains $c_r$ immediately followed by
	$c_{r+1}$, so under a one-VC assignment it contributes
	$(c_r,0)\rightarrow(c_{r+1},0)$. Taking all $r\in\mathbb Z_N$ yields
	the directed cycle~\eqref{eq:torus-one-vc-cycle}.
\end{proof}

\begin{proof}[Proof of Theorem~\ref{thm:torus}]
	\medskip\noindent\textbf{Minimal all-pairs connectivity.}
	For an ordered pair $([s],[t])$, let $s,t\in V_n$ be their canonical
	representatives and let $\lambda^\star=\bestlift(s,t)$. By
	Corollary~\ref{cor:unique-lift}, this is the unique closest lift vector.
	Lemma~\ref{lem:decomposition} gives
	\begin{equation}
		t-s+\lambda^\star=a\dir{i}+b\dir{i+1},
		\qquad a,b\geq0,
		\qquad a+b=d_{\torus}([s],[t]).
		\label{eq:torus-proof-decomposition}
	\end{equation}
	In the cross-group sector $S_5=\{\dir{5},\dir{0}\}$, the ordering
	$\dir{5}^{a}\dir{0}^{b}$ avoids both prohibited turns. In the
	cross-group sector $S_2=\{\dir{2},\dir{3}\}$, the ordering
	$\dir{3}^{b}\dir{2}^{a}$ does so. Every interleaving is permitted in
	the four within-group sectors. Projection preserves both endpoint and
	number of hops, so at least one permitted shortest route exists for
	every ordered quotient pair. Thus $\RT$ is minimal and all-pairs
	connected.
	
	More generally, every retained lifted geodesic uses at most two adjacent
	directions by Lemma~\ref{lem:geodesic-directions}. If
	$a,b>0$, all $\binom{a+b}{a}$ interleavings are retained in a
	within-group sector, whereas a cross-group sector admits only the unique
	$\mathrm L$-before-$\mathrm U$ word. If $a=0$ or $b=0$, the
	displacement lies on a sector boundary ray and the minimal direction
	word is unique.
	
	\medskip\noindent\textbf{Partial adaptivity.}
	For $n\geq3$, take $s=(0,0)$ and
	$t=\dir{0}+\dir{1}=(1,1)\in V_n$. The displacement has norm two.
	For every nonzero $\lambda\in\lattice$,
	Lemma~\ref{lem:minperiod} and the reverse triangle inequality give
	\begin{equation}
		\hexnorm{t+\lambda}
		\geq \hexnorm{\lambda}-\hexnorm{t}
		\geq (2n-1)-2 \geq 3.
		\label{eq:adaptivity-closest-lift}
	\end{equation}
	Thus $t$ is the unique closest lift. Both shortest direction words
	$\dir{0}\dir{1}$ and $\dir{1}\dir{0}$ belong to $\RT$ and offer
	distinct first channels, proving partial adaptivity for every $n\geq3$.
	
	For $n=2$, Corollary~\ref{cor:diameter} gives diameter one, so every
	nontrivial minimal route has exactly one hop. Moreover, two distinct
	outgoing directions cannot reach the same quotient neighbor: otherwise
	$\dir{i}-\dir{j}$ for $i\neq j$ would be a nonzero period vector of
	hexagonal norm at most two, contradicting
	Lemma~\ref{lem:minperiod}, which gives the lower bound three. Hence each
	destination has a unique minimal next channel, so $\RT$ is not partially
	adaptive for $n=2$. Therefore $\RT$ is partially adaptive if and only if
	$n\geq3$.
	
	\medskip\noindent\textbf{Single-dateline property.}
	Lemma~\ref{lem:hamiltonian} makes $H_G$ well-defined on quotient nodes.
	Consider a maximal same-group segment with channels
	$c_0,\ldots,c_{m-1}$. By~\eqref{eq:increments}, every hop has a positive
	integer increment $s_j\in\{1,k-1,k\}$ in the corresponding unwrapped
	Hamiltonian coordinate. Because the entire minimal route has at most
	$n-1$ hops by Corollary~\ref{cor:diameter},
	\begin{equation}
		0<\sum_{j=0}^{m-1}s_j
		\leq m k
		\leq(n-1)k
		=(n-1)(3n-1)=N-n<N.
		\label{eq:onecrossing}
	\end{equation}
	Reducing the strictly increasing unwrapped coordinate modulo $N$ can
	therefore pass from $N-1$ to $0$ at most once. Hence every maximal
	same-group segment contains at most one dateline channel, and the state
	\texttt{crossed} in~\eqref{eq:vc-map} changes from zero to one at most
	once within that segment.
	
	The same geodesic characterization determines the group order. Every
	retained path is confined to one adjacent-direction sector. A
	within-group sector lies wholly in one group, and either cross-group
	sector is forced by~\eqref{eq:bannedturns} to traverse its
	$\mathrm L$ directions before its $\mathrm U$ directions. Thus every
	resource-CDG edge is either intra-group or directed from $\mathrm L$ to
	$\mathrm U$.
	
	\medskip\noindent\textbf{Acyclicity of the complete resource CDG.}
	Define a rank only on resources that occur in $\RThat$. Let
	$c:[u]\rightarrow[v]$ be a used channel in group
	$G\in\{\mathrm L,\mathrm U\}$, let $h=H_G(u)\in\{0,\ldots,N-1\}$, and
	set $\ell(\mathrm L)=0$ and $\ell(\mathrm U)=1$. Define
	\begin{equation}
		r(c,q)=2N\ell(G)+
		\begin{cases}
			h, & q=0,\\
			h, & q=1\text{ and }c\text{ is a dateline channel for }G,\\
			N+h, & q=1\text{ and }c\text{ is not a dateline channel for }G.
		\end{cases}
		\label{eq:rank}
	\end{equation}
	This is a function of the resource $(c,q)$ alone. A dateline channel
	always receives $\VCone$ under~\eqref{eq:vc-map}; a non-dateline
	channel can receive $\VCone$ only after the unique dateline crossing of
	that same-group segment.
	
	We verify that every CDG edge strictly increases~\eqref{eq:rank}. By the
	definition in~\eqref{eq:cdg-edge}, it is enough to consider consecutive
	resources. Within a fixed group, let $s\in\{1,k-1,k\}$ be the positive
	unwrapped increment from the tail coordinate of the first channel to
	the tail coordinate of the next channel.
	\begin{enumerate}
		\item Before the dateline, both resources use $\VCzero$ and no modular
		wrap occurs, so the next tail coordinate is $h+s>h$ and the rank
		increases.
		\item If the second resource is the dateline channel, its tail is still
		the next pre-wrap coordinate $h+s>h$. The VC changes from $\VCzero$ to
		$\VCone$, but both rank cases use that pre-wrap tail coordinate, so the
		rank increases.
		\item If the first resource is the dateline channel and a subsequent
		channel exists, write the dateline tail coordinate as $h$ and its
		positive increment as $s$. Its head, and hence the next channel tail,
		has modular coordinate $h+s-N$. The next resource is a non-dateline
		$\VCone$ resource with rank contribution
		$N+(h+s-N)=h+s>h$.
		\item After the dateline, no second modular wrap is possible by
		\eqref{eq:onecrossing}. Consecutive non-dateline $\VCone$ resources both
		carry the offset $N$, while their modular tail coordinates increase by
		$s>0$; hence the rank again increases.
	\end{enumerate}
	Finally, every used resource in group $\mathrm L$ has rank at most
	$2N-1$, whereas every used resource in group $\mathrm U$ has rank at
	least $2N$. Therefore every cross-group dependency---including a reset
	from $\VCone$ in $\mathrm L$ to either VC on the first $\mathrm U$
	channel---strictly increases the rank.
	
	Every edge of the complete resource CDG of $\RThat$ strictly increases
	the integer-valued rank~\eqref{eq:rank}. A directed cycle would require
	the rank to return to its starting value after a sequence of strict
	increases, a contradiction. Hence the complete resource CDG is acyclic,
	and $\RThat$ is deadlock-free using two VCs.
\end{proof}

\begin{proof}[Proof of Corollary~\ref{cor:vc-transitions}]
	Every route has group word in $\mathrm L^*\mathrm U^*$, so it contains
	at most two maximal same-group segments. By the single-dateline property,
	the VC word within each such segment has the form $0^*1^*$. Concatenating
	the one or two segment words gives $0^*1^*$ or
	$0^*1^*0^*1^*$, which contains at most three VC-class transitions.
\end{proof}

\section{Proof of the Restricted Single-Link Fault Extension}
\label{app:single-link-fault}

We first record the rotational covariance used throughout the proof. Recall
the map $R_{60}(x,y)=(-y,x+y)$ from
Appendix~\ref{app:torusgeometry}, and define
\begin{equation}
	Q_j=R_{60}^{\,1-j}.
	\label{eq:fault-rotation}
\end{equation}

\begin{lemma}[Rotation covariance]
	\label{lem:fault-rotation}
	The map $R_{60}$ preserves $V_n$ and induces an automorphism of
	$\torus$. Moreover, $Q_j$ maps
	$\mathrm U_j$ to $\mathrm U$, maps $\mathrm L_j$ to $\mathrm L$, and
	maps the two turns in~\eqref{eq:fault-turns} to the two turns
	in~\eqref{eq:bannedturns}. Consequently, the rotated healthy relations
	$\mathcal R_{\mathrm M}^{(j)}$ and
	$\mathcal R_{\mathrm T}^{(j)}$ are minimal and all-pairs connected.
\end{lemma}

\begin{proof}
	Direct calculation gives
	$R_{60}(\dir{i})=\dir{i+1}$ for every $i$ and
	\begin{equation}
		\hexnorm{R_{60}(x,y)}
		=\max\{|y|,|x+y|,|x|\}
		=\hexnorm{(x,y)}.
		\label{eq:fault-rotation-norm}
	\end{equation}
	Hence $R_{60}$ preserves the vertex set $V_n$. For the torus,
	\begin{equation}
		R_{60}(T_1)=T_2,\qquad
		R_{60}(T_2)=T_2-T_1,
		\label{eq:fault-rotation-periods}
	\end{equation}
	so $R_{60}(\lattice)=\lattice$ and the map descends to a quotient
	automorphism.

	Because $Q_j(\dir{j+r})=\dir{1+r}$, it sends the three directions
	$\dir{j-1},\dir{j},\dir{j+1}$ to
	$\dir{0},\dir{1},\dir{2}$ and sends the opposite three directions to
	$\dir{3},\dir{4},\dir{5}$. It also sends
	$\dir{j+1}\rightarrow\dir{j+2}$ to
	$\dir{2}\rightarrow\dir{3}$ and sends
	$\dir{j-1}\rightarrow\dir{j-2}$ to
	$\dir{0}\rightarrow\dir{5}$. Thus each rotated healthy relation is the
	image of its original relation under a graph automorphism. The minimal
	all-pairs connectivity conclusions follow from
	Theorems~\ref{thm:mesh} and~\ref{thm:torus}.
\end{proof}

\begin{lemma}[Existence of a local triangle bypass]
	\label{lem:fault-bypass}
	Let $e=\{A,A+\dir{j}\}$ be a physical link of $\mesh$ or $\torus$,
	with $n\geq2$. On the torus, both two-hop substitutions for either
	direction of $e$ in~\eqref{eq:fault-bypass} exist. In the finite mesh, at
	least one of the two substitutions exists for either direction.
\end{lemma}

\begin{proof}
	The torus contains the quotient channel induced by every lattice step, so
	both intermediate quotient nodes and all four corresponding flank
	channels exist.

	For the mesh, Lemma~\ref{lem:fault-rotation} reduces the claim to
	$j=1$. Write
	\begin{equation}
		A=(x,y),\qquad B=A+\dir{1}=(x,y+1),
		\label{eq:fault-canonical-edge}
	\end{equation}
	and let $R=n-1\geq1$. The two candidate intermediate nodes are
	\begin{equation}
		C_0=A+\dir{0}=(x+1,y),\qquad
		C_2=A+\dir{2}=(x-1,y+1).
		\label{eq:fault-candidates}
	\end{equation}
	For $C_0$, the inequalities $|y|\leq R$ and
	$|x+y+1|\leq R$ follow from $A,B\in V_n$. Its only remaining condition
	is $|x+1|\leq R$, which can fail only when $x=R$. For $C_2$, the
	inequalities $|y+1|\leq R$ and $|x+y|\leq R$ again follow from
	$A,B\in V_n$. Its only remaining condition is $|x-1|\leq R$, which can
	fail only when $x=-R$. Since $R\geq1$, the two failures cannot occur
	simultaneously. Therefore at least one of $C_0,C_2$ lies in $V_n$.

	The forward paths through these nodes have direction words
	$\dir{0}\dir{2}$ and $\dir{2}\dir{0}$ because
	$\dir{0}+\dir{2}=\dir{1}$. Reversing them gives
	$\dir{5}\dir{3}$ and $\dir{3}\dir{5}$, and
	$\dir{3}+\dir{5}=\dir{4}$. Thus the same available intermediate node
	also supplies the reverse bypass. Rotating back proves the claim for
	every $j\in\{0,1,2\}$.
\end{proof}

\begin{lemma}[Properties of the fault transformation]
	\label{lem:fault-transformation}
	For every $P\in\mathcal R_X^{(j)}(s,t)$,
	$X\in\{\mathrm M,\mathrm T\}$, the set $\mathcal B_e(P)$ is nonempty.
	Every route $P_e\in\mathcal B_e(P)$ avoids the failed link, satisfies
	\begin{equation}
		|P_e|\leq |P|+1,
		\label{eq:fault-length}
	\end{equation}
	and has group word in $\mathrm L_j^*\mathrm U_j^*$.
\end{lemma}

\begin{proof}
	By Lemma~\ref{lem:fault-rotation}, $P$ is a shortest path in the healthy
	graph and is therefore simple. In particular, it traverses the physical
	link $e$ at most once. If it avoids $e$, then
	$\mathcal B_e(P)=\{P\}$ and all claims follow immediately.

	Otherwise, Lemma~\ref{lem:fault-bypass} provides at least one available
	substitution. The vector identities
	\begin{equation}
		\dir{j-1}+\dir{j+1}=\dir{j},\qquad
		\dir{j+2}+\dir{j+4}=\dir{j+3}
		\label{eq:fault-vector-identities}
	\end{equation}
	show that every substitution has the same endpoints as the failed hop.
	Its two directions differ from the failed direction and its reverse, so
	the resulting route avoids $e$. Exactly one hop is replaced by two,
	which gives equality in~\eqref{eq:fault-length} for an affected route.
	No U-turn is introduced: a healthy geodesic containing $\dir{j}$ uses
	only $\dir{j}$ and at most one adjacent direction, while neither bypass
	flank is opposite to a direction in either adjacent sector. The reverse
	case is symmetric.

	The substituted route is also a path rather than a route with a repeated
	vertex. If a repetition were introduced, the repeated portion would be a
	closed subwalk. It cannot have length one because the graphs have no
	self-loops, and it cannot have length two because the transformed
	direction word has no U-turn. Removing a closed portion of length at
	least three from a route of length $|P|+1$ would produce a healthy
	$s$--$t$ walk shorter than $|P|$, contradicting the minimality of $P$.

	The rotated turn restrictions give $P$ a group word in
	$\mathrm L_j^*\mathrm U_j^*$. A $\dir{j}$ hop belongs to
	$\mathrm U_j$ and is replaced by two $\mathrm U_j$ hops, while a
	$\dir{j+3}$ hop belongs to $\mathrm L_j$ and is replaced by two
	$\mathrm L_j$ hops. The substitution therefore preserves the group word
	language.
\end{proof}

\begin{proof}[Proof of Proposition~\ref{prop:single-link-fault}]
	\medskip\noindent\textbf{Connectivity and path length.}
	Fix an ordered source--destination pair. By
	Lemma~\ref{lem:fault-rotation}, its rotated healthy relation contains a
	shortest path $P$. Lemma~\ref{lem:fault-transformation} makes
	$\mathcal B_e(P)$ nonempty, and every member is a valid route in the
	faulted graph. Thus both fault-aware relations are all-pairs connected.
	The healthy path has length $d_{\mathrm H}(s,t)$ in the mesh and
	$d_{\torus}([s],[t])$ in the torus. \Eqref{eq:fault-length}
	gives the stated additive one-hop bound for every permitted fault-aware
	route. This is a bound relative to the healthy shortest-path distance,
	not a claim that every fault-aware route is shortest in the faulted
	graph.

	\medskip\noindent\textbf{Mesh deadlock freedom.}
	Let $\Phi(x,y)=x+2y$ be the potential used in the proof of
	Theorem~\ref{thm:mesh}, and define
	\begin{equation}
		\Phi_j(v)=\Phi(Q_jv).
		\label{eq:fault-potential}
	\end{equation}
	The map $Q_j$ sends the ordered directions of $\mathrm U_j$ to
	$\dir{0},\dir{1},\dir{2}$, whose $\Phi$ increments are $1,2,1$,
	and sends the ordered directions of $\mathrm L_j$ to
	$\dir{3},\dir{4},\dir{5}$, whose increments are $-1,-2,-1$.
	Consequently, $\Phi_j$ strictly increases along every
	$\mathrm U_j$ channel and strictly decreases along every
	$\mathrm L_j$ channel, including all channels inserted by either bypass
	order.

	By Lemma~\ref{lem:fault-transformation}, every fault-aware group word
	belongs to $\mathrm L_j^*\mathrm U_j^*$. Hence every edge of the complete
	one-VC CDG is either within one group or directed from
	$\mathrm L_j$ to $\mathrm U_j$. A directed CDG cycle using both groups
	would require a dependency from $\mathrm U_j$ back to $\mathrm L_j$,
	which does not exist. A cycle wholly within $\mathrm U_j$ would force
	$\Phi_j$ to increase strictly around a closed dependency chain, and a
	cycle wholly within $\mathrm L_j$ would force it to decrease strictly.
	Both are impossible. The mesh complete resource CDG is therefore
	acyclic using one VC.

	\medskip\noindent\textbf{Rotated torus coordinates.}
	For $G\in\{\mathrm L_j,\mathrm U_j\}$, write
	\begin{equation}
		H_{G,j}=
		\begin{cases}
			H_{\mathrm L,j},&G=\mathrm L_j,\\
			H_{\mathrm U,j},&G=\mathrm U_j.
		\end{cases}
		\label{eq:fault-group-coordinate}
	\end{equation}
	Lemma~\ref{lem:fault-rotation} makes $Q_j$ a quotient automorphism, so
	the maps in~\eqref{eq:fault-hamcoords} are well-defined bijections
	because $\HU$ and $\HL$ are well-defined bijections. Their positive
	unwrapped direction increments are
	\begin{equation}
		\begin{array}{c|ccc}
			& \text{first direction} & \text{middle direction} & \text{third direction}\\ \hline
			\mathrm U_j
			& \Delta H_{\mathrm U,j}(\dir{j-1})=1
			& \Delta H_{\mathrm U,j}(\dir{j})=k
			& \Delta H_{\mathrm U,j}(\dir{j+1})=k-1\\
			\mathrm L_j
			& \Delta H_{\mathrm L,j}(\dir{j+2})=1
			& \Delta H_{\mathrm L,j}(\dir{j+3})=k
			& \Delta H_{\mathrm L,j}(\dir{j+4})=k-1.
		\end{array}
		\label{eq:fault-increments}
	\end{equation}

	\medskip\noindent\textbf{Preservation of the single-dateline property.}
	In either group, replacing a middle-direction hop by either order of the
	two flank directions preserves its exact unwrapped advance:
	\begin{equation}
		1+(k-1)=(k-1)+1=k.
		\label{eq:fault-advance}
	\end{equation}
	The substitution does not change group boundaries. Each maximal
	same-group segment of a fault-aware route therefore corresponds to a
	segment of its healthy parent route and has exactly the same total
	unwrapped Hamiltonian advance. After rotation, the healthy bound
	in~\eqref{eq:onecrossing} remains
	\begin{equation}
		0<S\leq(n-1)k=N-n<N.
		\label{eq:fault-onecrossing}
	\end{equation}
	Although an affected physical route has one additional hop, its value of
	$S$ is unchanged. Since every individual increment
	in~\eqref{eq:fault-increments} is positive, reducing the unwrapped
	coordinate modulo $N$ still produces at most one wrap in each maximal
	same-group segment. Thus the rotated per-hop VC rule is well-defined and
	changes from its pre-dateline to post-dateline phase at most once per
	segment.

	\medskip\noindent\textbf{Torus resource-CDG acyclicity.}
	Let $c:[u]\rightarrow[v]$ be a channel used by the fault-aware relation
	in group $G$, let $h=H_{G,j}(u)\in\{0,\ldots,N-1\}$, and set
	$\ell_j(\mathrm L_j)=0$ and $\ell_j(\mathrm U_j)=1$. Define
	\begin{equation}
		r_j(c,q)=2N\ell_j(G)+
		\begin{cases}
			h, & q=0,\\
			h, & q=1\text{ and }c\text{ is a dateline channel for }G,\\
			N+h, & q=1\text{ and }c\text{ is non-dateline}.
		\end{cases}
		\label{eq:fault-rank}
	\end{equation}
	This rank depends only on the resource $(c,q)$.

	Consider consecutive resources within one group, and let
	$s\in\{1,k-1,k\}$ be the positive increment of the first channel. Before
	the dateline, a $\VCzero\rightarrow\VCzero$ dependency changes the tail
	coordinate from $h$ to $h+s$ without wrap and strictly increases the
	rank. If the second resource is the dateline channel, its tail coordinate
	is again $h+s>h$; assigning it $\VCone$ uses the dateline case of
	\eqref{eq:fault-rank} and still increases the rank. If the first resource
	is the dateline channel, the next tail coordinate is $h+s-N$, and the
	next non-dateline $\VCone$ resource has contribution
	$N+(h+s-N)=h+s>h$. After that crossing,
	\Eqref{eq:fault-onecrossing} excludes a second wrap, so two
	consecutive non-dateline $\VCone$ resources retain the $N$ offset while
	their tail coordinates increase by $s$. These cases include the new
	bypass turns because the argument uses only their group membership and
	the positive increment of the first channel.

	Finally, every used $\mathrm L_j$ resource has rank at most $2N-1$,
	whereas every used $\mathrm U_j$ resource has rank at least $2N$.
	Lemma~\ref{lem:fault-transformation} permits only
	$\mathrm L_j\rightarrow\mathrm U_j$ cross-group dependencies, so each
	such dependency strictly increases the rank even when the VC state is
	reset. Every edge of the complete resource CDG therefore strictly
	increases~\eqref{eq:fault-rank}. A directed cycle is impossible, and the
	torus fault-aware relation is deadlock-free using two VCs.
\end{proof}

The proposition is limited to a single static bidirectional link failure
under a globally consistent routing epoch. With multiple failed links, a
triangle bypass may itself be unavailable; a node failure removes several
incident channels; and an in-flight configuration change permits
dependencies from both the old and new relations. None of those enlarged
resource relations is covered by the proof above.

\section{Reassessment of the Three-VC Scheme of Shamaei et al.}
\label{app:prior-reassessment}

This appendix verifies Proposition~\ref{prop:prior-cycle} under the routing
and VC-assignment rules stated by \citet{shamaei2013hex}. A Type-1 route may
take an arbitrary ordering of $a$ steps in $\omega^0$ and $b$ steps in
$\omega^1$; a message is classified as wraparound when it traverses at least
one wraparound link; and every hop of a wraparound Type-1 message uses VC
class~1. Under the identification $\omega^0=\dir{0}$ and
$\omega^1=\dir{1}$, the generator $4+3\omega$ yields the period vectors
$T_1=(4,3)$ and $T_2=(-3,7)$ used to represent their $H_4$ instance in our coordinates.

Table~\ref{tab:prior-witness} gives one complete route witnessing each edge
of the cycle in~\eqref{eq:prior-cycle}. A superscript $W$ marks a projected
hop that crosses the chosen fundamental-domain boundary; unmarked hops are
regular. Coordinates are canonical representatives in $V_4$.

\begin{table}[H]
	\caption{Minimal wraparound Type-1 routes witnessing the seven edges
		of~\eqref{eq:prior-cycle}.}
	\label{tab:prior-witness}
	\centering
	\small
	\resizebox{\linewidth}{!}{%
		\begin{tabular}{cclc}
			\toprule
			Packet & Dependency & Complete projected route & Length\\
			\midrule
			$P_0$ & $c_0\to c_1$ &
			$(3,0)\xrightarrow{W}(0,-3)\to(1,-3)$ & 2\\
			$P_1$ & $c_1\to c_2$ &
			$(-3,3)\xrightarrow{W}(0,-3)\to(1,-3)\to(2,-3)$ & 3\\
			$P_2$ & $c_2\to c_3$ &
			$(1,-3)\to(2,-3)\to(3,-3)\xrightarrow{W}(-3,1)$ & 3\\
			$P_3$ & $c_3\to c_4$ &
			$(-1,3)\xrightarrow{W}(2,-3)\to(3,-3)\to(3,-2)$ & 3\\
			$P_4$ & $c_4\to c_5$ &
			$(3,-3)\to(3,-2)\to(3,-1)\xrightarrow{W}(-3,3)$ & 3\\
			$P_5$ & $c_5\to c_6$ &
			$(3,-2)\to(3,-1)\to(3,0)\xrightarrow{W}(0,-3)$ & 3\\
			$P_6$ & $c_6\to c_0$ &
			$(3,-1)\to(3,0)\xrightarrow{W}(0,-3)$ & 2\\
			\bottomrule
		\end{tabular}%
	}
\end{table}

\begin{figure}[!tbp]
	\centering
	\includegraphics[width=0.90\linewidth]{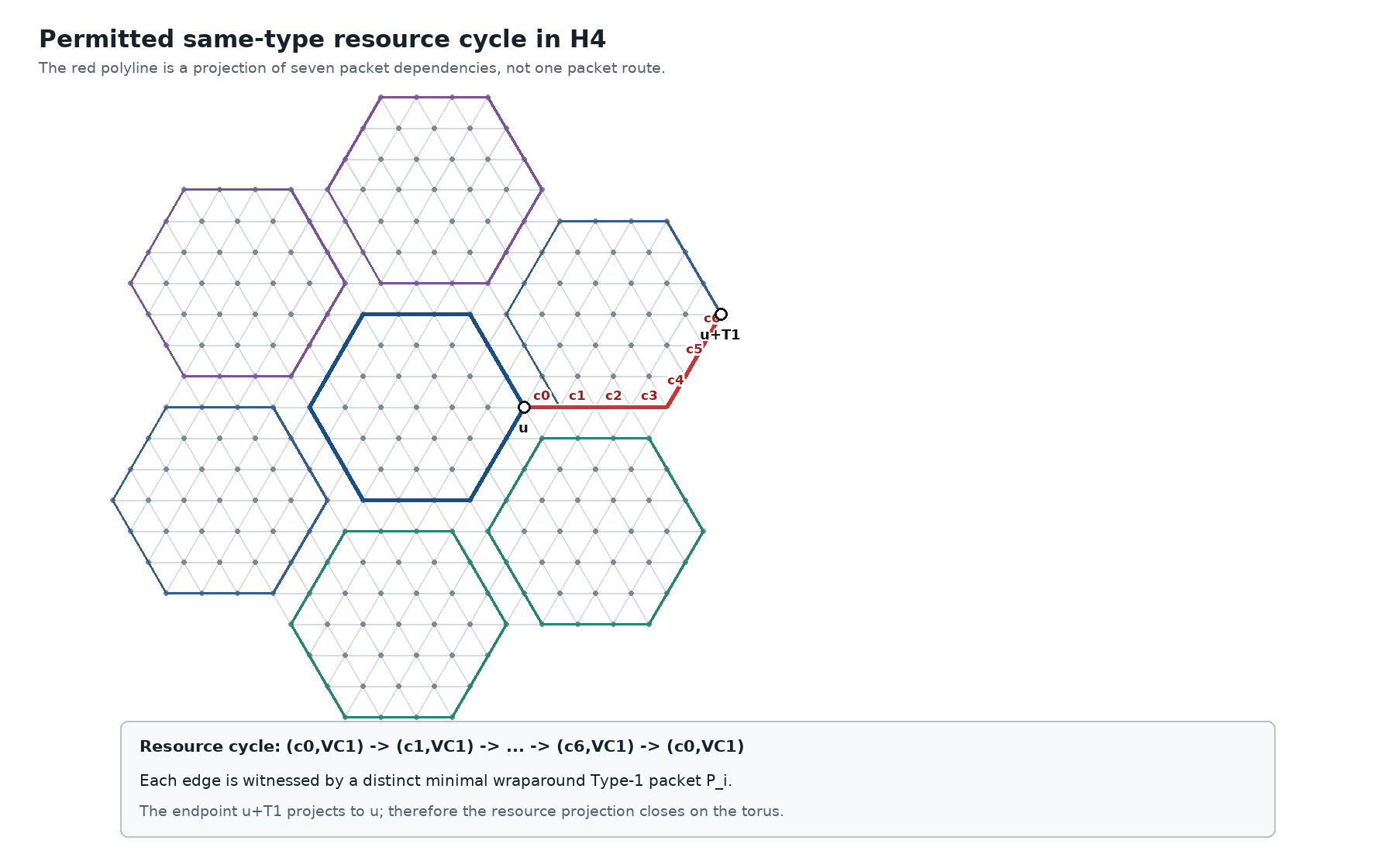}
	\caption{Unfolded visualization of the resource cycle in
		Proposition~\ref{prop:prior-cycle}. The red polyline is not a single
		seven-hop packet route. Segment $c_i$ is a resource held by witness
		packet $P_i$ while that packet requests $c_{i+1}$ (indices modulo
		seven). The two endpoints differ by $T_1$ and therefore project to the
		same torus node.}
	\label{fig:prior-counterexample}
\end{figure}

\begin{proof}[Verification of Proposition~\ref{prop:prior-cycle}]
	Interpret the projected channels in Table~\ref{tab:prior-witness} modulo
	\begin{equation}
		\Lambda_4=\langle T_1,T_2\rangle,
		\qquad T_1=(4,3),\quad T_2=(-3,7).
		\label{eq:lambda4}
	\end{equation}
	Every listed hop is then in direction $\dir{0}$ or $\dir{1}$, so every
	witness is a Type-1 route. Each row contains at least one marked
	wraparound hop; under the published VC table, every resource used by that
	packet is therefore in VC class~1.
	
	It remains to verify minimality without relying on the claimed acyclicity
	result. Lift each projected route to the infinite lattice using its
	$\dir{0}/\dir{1}$ direction word, and let $\Delta_i$ be its lifted
	source-to-destination displacement. Because
	$\Delta_i=a_i\dir{0}+b_i\dir{1}$ with $a_i,b_i\geq0$, Lemmas
	\ref{lem:decomposition} and~\ref{lem:lattice-distance} give
	\begin{equation}
		\hexnorm{\Delta_i}=a_i+b_i=L_i,
		\qquad L_i\in\{2,3\},
		\label{eq:prior-witness-length}
	\end{equation}
	where $L_i$ is the listed route length. Lemma~\ref{lem:minperiod} with
	$n=4$ gives $\hexnorm{\mu}\geq7$ for every nonzero
	$\mu\in\Lambda_4$. Hence, for every alternative destination lift,
	\begin{equation}
		\hexnorm{\Delta_i+\mu}
		\geq \hexnorm{\mu}-\hexnorm{\Delta_i}
		\geq 7-L_i>L_i.
		\label{eq:prior-minimality}
	\end{equation}
	Thus the displayed lift is the unique closest lift and every witness
	route is minimal on the $H_4$ quotient represented above.
	
	For each $i$, the complete route $P_i$ contains channel $c_i$ immediately
	followed by $c_{i+1}$, with indices modulo seven. Since both resources
	use VC class~1, $P_i$ induces the CDG edge
	$(c_i,1)\rightarrow(c_{i+1},1)$. Taking all seven witness packets yields
	the closed directed resource cycle~\eqref{eq:prior-cycle}.
\end{proof}

The argument above establishes only the statement of
Proposition~\ref{prop:prior-cycle}. Theorem~2 of the prior paper reasons from
the length $2n-1$ of a complete same-type geometric cycle relative to the
network diameter $n-1$. That observation does not exclude a resource-CDG
cycle assembled from distinct minimal packets, each contributing only one
adjacent resource dependency. The witnesses in Table~\ref{tab:prior-witness}
have length at most the $H_4$ diameter of three hops, yet their union contains
the seven-resource cycle in~\eqref{eq:prior-cycle}. Therefore the stated
argument is insufficient to establish acyclicity of the complete resource
CDG.

This conclusion is deliberately limited. For a fully adaptive routing
function, a cycle in the union CDG is not by itself a construction of a
reachable wormhole deadlock: packets may retain alternative outputs or an
escape structure not represented by one selected witness route. The result
therefore contradicts the published complete-CDG acyclicity claim, but it
does not by itself prove that the full three-VC routing algorithm can reach a
packet-level deadlock.

\end{document}